\documentclass{LMCS}
\usepackage{epsfig}
\usepackage{amsmath, amsfonts, amssymb, latexsym, stmaryrd}
\usepackage{amscd}

\usepackage{enumerate}
\usepackage{hyperref}

\usepackage{graphicx,color}
\usepackage[all]{xy}

\usepackage[T1]{fontenc}
\usepackage{array}

\usepackage{pgf}
\usepackage{tikz}

\usetikzlibrary{snakes}
\usepgflibrary{snakes}
\usepackage{pgflibraryshapes}
\newcommand{\forget}[1]{}
\newcommand{\realnum}{{r}}  
\newcommand{\reals}{{\mathbb R}}
\newcommand{\realsnull}{\reals_{\geq 0}}
\newcommand{\Nat}{{\mathbb N}}

\newcommand{\Eff}{{\mathcal E}}       
\renewcommand{\SS}{\hat{S}}          

\newcommand{\semantics}[2]{\llbracket #1 \rrbracket_{#2}}
\newcommand{\unpsemantics}[1]{\semantics{#1}{}}
\newcommand{\dsemantics}[1]{\semantics{#1}{B}}

\newcommand{\densesemantics}[1]{\semantics{#1}{\realsnull}}
\newcommand{\sampledsemantics}[1]{\semantics{#1}{\epsilon}}
\newcommand{\actions}{\Sigma}

\newcommand{\fract}{{\sf fr}}
\newcommand{\intg}{{\sf int}}
\newcommand{\clocks}{{\mathcal C}}
\newcommand{\arrow}[1]{\stackrel{#1}\longrightarrow}

\newcommand{\Dom}{{\mathbb T}}        

\newcommand{\ppr}{{\rightarrow_p}}    
\newcommand{\Null}{{\sf null}}        
\newcommand{\rank}{{\sf rank}}        

\newcommand{\numodels}[1]{\overline{#1}_{\nu}}                
\newcommand{\nupmodels}[1]{\overline{#1}_{\nu'}}              

\newcommand{\Dmodels}[3]{#1#2_D#3}        
\newcommand{\Dpmodels}[3]{#1#2_{D'}#3}      

\newcommand{\epsmodels}[2]{#1 \models_\epsilon #2}        
\newcommand{\twoepsmodels}[2]{#1 \models_{2\epsilon} #2}        

\newcommand{\plus}{\oplus}             

\theoremstyle{plain}\newtheorem{lemma}[thm]{Lemma}

\newcommand{\note}[1]{\forget{#1}}

\renewcommand{\revision}[1]{{\forget{#1}}}

\def\doi{6 (3:14) 2010}
\lmcsheading%
{\doi}
{1--37}
{}
{}
{May\phantom.~10, 2009}
{Sep.~\phantom01, 2010}
{}

\begin{document}

\title[Sampled Semantics of Timed Automata]{Sampled Semantics of Timed Automata}

\author[P.~A.~Abdulla]{Parosh Aziz Abdulla}	
\address{Department of Information Technology, Uppsala University, Sweden}	 \email{\{parosh,pavelk,yi\}@it.uu.se}  

\author[P.~Krcal]{Pavel Krcal}	

\author[W.~Yi]{Wang Yi}	






\begin{abstract}
  Sampled semantics of timed automata is a finite approximation of their dense time behavior. While the former is closer to the actual software or hardware systems with a fixed granularity of time, the abstract character of the latter makes it appealing for system modeling and verification. We study one aspect of the relation between these two semantics, namely checking whether the system exhibits some qualitative (untimed) behaviors in the dense time which cannot be reproduced by any implementation with a fixed sampling rate. More formally, the \emph{sampling problem} is to decide whether there is a sampling rate such that all qualitative behaviors (the untimed language) accepted by a given timed automaton in dense time semantics can be also accepted in sampled semantics. We show that this problem is decidable.
\end{abstract}

\maketitle

\section{Introduction}
\label{Sec:Introduction}

\note{Notational conventions: $n$ is the number of counters, $D$ is a region, $B$ is a bound. $\delta$ is the label for time successor, $t$ is a letter for effects. We use the term 'instruction' for $0,1,r,*j$, 'effect' for $n$-tuples of instructions, 'operation' is a synonym for 'function'.}
\note{Put somewhere that a part of this paper has been published as~\cite{aky08rc-automata}?}

Dense time semantics allows timed automata~\cite{alur94theory} to delay for arbitrary real valued amounts of time. This includes also arbitrarily small delays and delays which differ from each other by arbitrarily small values. Neither of these behaviors can be enforced by an implementation operating on a concrete hardware. Each such implementation necessarily includes some (hardware) digital clock which determines the least time delay measurable or enforceable by the system.

This observation motivates \emph{sampled semantics} of timed automata, which is a discrete time semantics with the smallest time step fixed to some fraction of $1$. In other words, the time delays in a sampled semantics with the smallest step $\epsilon$ can be only multiples of $\epsilon$. There are infinitely many different sampled semantics, but any of them allows fewer behaviors of the system than dense time semantics. On the other hand, all of the allowed behaviors in a sampled semantics with the sampling rate (the smallest step) $\epsilon$ will be preserved in an implementation on a platform with the clock rate $\epsilon$ (and all fractions of $\epsilon$).

One of the arguments in favor of using dense time semantics is that one does not have to consider a concrete sampling rate of an implementation in the modeling and analysis phase. Dense time semantics abstracts away from concrete sampling rates by including all of them. Also, it seems adequate to assume that the environment stimuli come at any real time point without our control. 

If a concrete timed automaton serves as a system description for later implementation, one might try to find a sampling rate which preserves all qualitative behaviors (untimed words). The restriction to qualitative behaviors is necessary, because any sampling rate excludes infinitely many dense time behaviors. By this we lose the explicit timing information, but many important properties, including implicit timing, are preserved. For instance, if we know that the letter $b$ cannot appear later than $5$ time units after an occurrence of the letter $a$ in the dense time model and then there is an untimed word accepted by this automaton where $a$ is followed by $b$ then we know that there is a run where $b$ comes within $5$ time units after $a$.

The problem of our interest can be formalized as follows: decide whether for a given timed automaton there is a sampling rate such that all untimed words accepted by the automaton in dense time semantics are also accepted in sampled semantics with the fixed sampling rate. We call this the \emph{sampling problem} for timed automata. 

There are timed automata with qualitative behaviors which are not accepted in any sampled semantics. This relies on the fact that timed automata can force differences between the fractional parts of the clock values to grow. In sampled semantics with the smallest time step fixed to $\epsilon$, the distance can only be increased in multiples of $\epsilon$, which implies that the distance between a pair of clocks can grow at most $1/\epsilon$ times. One more increase would make the fractional parts equal again. A sampling rate ensuring acceptance of an untimed word must induce enough valuations within each clock region in order to accommodate increases of the distances between the fractional parts of clock values along some accepting run. If there is a sequence of untimed words which require smaller and smaller time steps in order to be accepted then any fixed sampling necessarily loses some of these words.

To enforce clock difference growth, a timed automaton has to use strict inequalities $<$ and $>$ in its clock guards. Closed timed automata, i.e., timed automata with only non-strict inequalities $\leq$ and $\geq$ in the guards, can be always sampled with the sampling rate $1$. Closed timed automata possess one important property -- they are {\it closed under digitization}~\cite{ouaknine03universality}. The property {\it "closed under digitization"} has been defined in~\cite{henzinger92what} and it is connected to our problem in the following sense: if the timed language of a timed automaton is closed under digitization then all (untimed) behaviors of this timed automaton are preserved with $\epsilon = 1$. Also, closure under digitization was shown to be decidable in~\cite{ow-lics03}.  
\note{\cite{bmt99efficient} states this in a bit strange way. It is not clear to me if it is really the same thing, but they refer to the same sources.}

The growth of clock value differences corresponds to a special type of memory. When a clock value difference grows three times then there must be at least three different clock value differences smaller than the current one. We show that this memory can be characterized by a new type of counter automata -- with finite state control and a finite number of unbounded counters taking values from the natural numbers. The counters can be updated along the transitions by the following instructions:
\begin{enumerate}[$\bullet$]
  \item $0$: the counter keeps its value unchanged,
  \item $1$: the counter value is incremented,
  \item $r$: the counter value is reset to $0$,
  \item copy: the counter value is set to the value of another counter,
  \item $\max$: under some conditions, the counter value can be set to the maximum of sums of pairs of counters.
\end{enumerate}

\noindent The sampling problem can be reformulated for our counter automata as follows. We want to decide whether there is a bound such that all words accepted by the automaton can be accepted also by runs along which all counters are bounded by this bound. This problem was studied earlier as the \emph{limitedness problem} for various types of automata with counters. We show that this problem is decidable for our automata by reducing it to the limitedness problem of a simpler type of automata, R-automata~\cite{aky08r-automata}.

\paragraph{\bf Related work.} The problem of asking for a sampling rate which satisfies given desirable properties has been studied in~\cite{amp98disc,cassez02comparison,kp05on}. In~\cite{amp98disc}, the authors identify subclasses of timed automata (or, digital circuits which can be translated to timed automata) such that there is always an $\epsilon$ which preserves all qualitative behaviors. The problem of deciding whether there is a sampling rate ensuring  language non-emptiness is studied in~\cite{cassez02comparison,kp05on}.
Work on digitization of timed languages~\cite{henzinger92what} identifies systems for which verification results obtained in discrete time transfer also to the dense time setting. Digitization takes timing properties into account more explicitly, while we consider only qualitative behaviors. A different approach to discretization has been developed in~\cite{GPV94}. This discretization scheme preserves all qualitative behaviors for the price of skewing the time passage.
Implementability of systems modeled by timed automata on a digital hardware has been  studied in~\cite{wdr04aasap,kmty-concur04,altisen05implementation}. The papers~\cite{wdr04aasap,kmty-concur04} propose a new semantics of timed automata with which one can implement a given system on a sufficiently fast platform. On the other hand,~\cite{altisen05implementation} suggests a methodology in which the hardware platform is modeled by timed automata in order to allow checking whether the system satisfies the required properties on the given platform.

The limitedness problem has been studied for various types of finite automata with counters. First, it has been introduced by Hashiguchi~\cite{hashiguchi82limitedness} for distance automata (automata with one counter which can be only incremented). Different proofs of the decidability of the limitedness problem for distance automata are reported in~\cite{hashiguchi90improved,leung91limitedness,simon94semigroups}. Distance automata were extended in~\cite{kirsten05distance} with additional counters which can be reset following a hierarchical discipline resembling parity acceptance conditions. Our automata relax this discipline and allow the counters to be reset arbitrarily. Universality of a similar type of automata for tree languages is studied in~\cite{ICALP08:colcombet-loeding-mostowski,CSL08:colcombet-loeding-depth-mu-calculus}. A model with counters which can be incremented and reset in the same way as in R-automata, called B-automata, is presented in~\cite{bc06bounds}. B-automata accept infinite words such that the counters are bounded along an infinite accepting computation.

\paragraph{\bf Structure of the Paper.} The rest of the paper is organized as follows. In Section~\ref{Sec:Preliminaries}, we introduce timed automata, dense time and sampled semantics, and our problem. Moreover, we define some technical concepts. Section~\ref{Sec:Results} states the result and sketches the structure of the proof. The model of automata with counters is presented in Section~\ref{Sec:extended-R-automata}, where also the important properties of these automata are shown. The main step of the proof, the construction of a counter automaton from a given timed automaton, together with the correspondence proofs is in Section~\ref{Sec:Encoding-TA-RC}. The proof is completed in Section~\ref{Sec:proof}.

\section{Preliminaries}
\label{Sec:Preliminaries}

In this section, we define syntax and two types of semantics (standard real time and sampled semantics) of timed automata and our problem. We also define region graphs for timed automata and a new notation which simplifies talking about clock differences and clock regions. Let $\Nat$ denote the set of non-negative integers.

\paragraph{\bf Syntax.} Let $\clocks$ be a finite set of non-negative real-valued variables called \emph{clocks}. The set of guards $G(\clocks)$ is defined by the grammar $g := x \bowtie c \mid g \wedge g$ where $x \in \clocks, c \in \Nat$ and $\bowtie\ \in \{ <, \leq, \geq, > \}$. A \emph{timed automaton} is a tuple $A = (Q, \actions, \clocks, q_0, E, F)$, where:
\begin{enumerate}[$\bullet$]
\item $Q$ is a finite set of locations,
\item $\Sigma$ is a finite alphabet,
\item $\clocks$ is a finite set of clocks,
\item $q_0 \in Q$ is an initial location,
\item $E \subseteq Q \times \actions \times G(\clocks) \times 2^\clocks \times Q$ is a finite transition relation, and
\item $F\subseteq Q$ is a set of accepting locations.
\end{enumerate}


\paragraph{\bf Semantics.}
Semantics is defined with respect to a given time domain $\Dom$. We suppose that a time domain is a subset of real numbers which contains $0$ and is closed under addition. Also, we suppose that $\Dom \cap \actions = \emptyset$. A \emph{clock valuation} is a function $\nu: \clocks \rightarrow \Dom$. If $\realnum \in \Dom$ then a valuation $\nu + \realnum$ is such that for each clock $x \in \clocks$, $(\nu + \realnum)(x) = \nu(x) + \realnum$. If $Y \subseteq \clocks$ then a valuation $\nu[Y:=0]$ is such that for each clock $x \in \clocks \smallsetminus Y$, $\nu[Y:=0](x) = \nu(x)$ and for each clock $x \in Y$, $\nu[Y:=0](x) = 0$. The satisfaction relation $\nu \models g$ for $g \in G(\clocks)$ is defined in the natural way.

The semantics of a timed automaton $A = (Q, \actions, \clocks, q_0, E, F)$ with respect to the time domain $\Dom$ is a labeled transition system (LTS) $ \semantics{A}{\Dom} = (\hat{Q}, \actions \cup \Dom, \rightarrow, \hat{q_0})$ where $\hat{Q} = Q \times \Dom^\clocks$ is the set of states, $\hat{q_0} = \langle q_0, \nu_0 \rangle$ is the initial state, $\nu_0(x) = 0$ for all $x \in \clocks$. The transition relation is defined as follows: $\langle q, \nu \rangle \arrow{\alpha} \langle q', \nu' \rangle$ if and only if
\begin{enumerate}[$\bullet$]
\item {\bf time step:} $\alpha \in \Dom$, $q = q'$, and $\nu' = \nu + \alpha$, or
\item {\bf discrete step:} $\alpha \in \actions$, there is $(q,a, g,Y,q') \in  E$, $\nu \models g, \nu' = \nu[Y:=0]$.

\end{enumerate}

We call paths in the semantics LTS \emph{runs}. For a finite run $\rho$ let $l(\rho)\in (\actions \cup \Dom)^*$ be the sequence of labels along this path. Let $l(\rho)\upharpoonright \Dom \in \actions^*$ be the sequence of labels with all numbers projected out. We use the same notation also for infinite (countable) runs containing infinitely many discrete steps. Namely, $l(\rho)\upharpoonright \Dom \in \actions^\omega$ if $\rho$ is such a run.

\paragraph{\bf Language.} A finite run $\rho = \langle q_0, \nu_0 \rangle \arrow{}^* \langle q, \nu \rangle$ is accepting if $q \in F$. The (untimed finite word) language of a timed automaton $A$ parameterized with the time domain $\Dom$, denoted $L_{\Dom}(A)$ is the set of words which can be read along the accepting runs of the semantics LTS.  Formally, $L_{\Dom}(A) = \{l(\rho)\upharpoonright \Dom\ |\ \rho$ is a finite accepting run in $\semantics{A}{\Dom} \}$.

An infinite (countable) run with infinitely many discrete steps is accepting if it contains an infinite set of states $\{\langle q,\nu_i \rangle | i\in\Nat \}$ such that $q\in F$ (standard B\"{u}chi acceptance condition). The (untimed) $\omega$-language of a timed automaton $A$ parameterized with the time domain $\Dom$, denoted $L_{\Dom}^\omega(A)$ is the set of words which can be read along the infinite countable accepting runs of the semantics LTS.  Formally, $L_{\Dom}^\omega(A) = \{l(\rho)\upharpoonright \Dom\ |\ \rho$ is an infinite countable accepting run in $\semantics{A}{\Dom} \}$.

Let $\realsnull$ be the set of all non-negative real numbers. Let the time domain $\Dom_\epsilon$ for an $\epsilon = 1/k$ for some $k \in \Nat$ be the set $\Dom_\epsilon = \{ l\cdot \epsilon \mid l\in \Nat\}$. We consider the time domains $\realsnull$ and $\Dom_\epsilon$ for all $\epsilon$. The semantics induced by $\realsnull$ is called \emph{dense time semantics} and the semantics induced by a $T_\epsilon$ is called \emph{$\epsilon$-sampled semantics}. We  use the following shortcut notation: $\sampledsemantics{A} = \semantics{A}{\Dom_\epsilon}, L(A) = L_{\realsnull}(A), L^\omega(A) = L^\omega_{\realsnull}(A), L_\epsilon(A) = L_{\Dom_\epsilon}(A), L_\epsilon^\omega(A) = L^\omega_{\Dom_{\epsilon}}(A)$.

\paragraph{\bf Problems.} We deal with the following problems. Decide for a timed automaton $A$ whether there is an $\epsilon = 1/k$ for some $k\in\Nat$ such that
\begin{enumerate}[$\bullet$]
  \item $L_\epsilon(A) = L(A)$, (sampling)
  \item $L^\omega_\epsilon(A) = L^\omega(A)$ ($\omega$-sampling).
\end{enumerate}

There are timed automata such that no matter how small $\epsilon$ we choose, $L_\epsilon(A) \neq L(A)$ and (or) $L_\epsilon^\omega(A) \neq L^\omega(A)$. As an example, consider the timed automaton in Figure~\ref{Fig:no-sampling}. It enforces the difference between clock values to shrink while being strictly greater than $0$. If the values of $x,y$ are $0.1, 0.6$, respectively, in the location $q_1$ then the difference between the clock values in the location $q_1$ after reading $ba$ will be strictly smaller than $0.5$.

\begin{figure}[htbp]
\centering
\begin{tikzpicture}[>=stealth]

\draw(0,0) node [circle,draw,name=s0]{$q_0$};
\draw[xshift= 4cm](0,0) node [circle,double,draw,name=s1]{$q_1$};
\draw[xshift= 8cm](0,0) node [circle,draw,name=s2]{$q_2$};

\draw[->] (s0) -- (s1) node[midway,above]{$a, x<1 \wedge y<1$}
                       node[midway,below]{$x:=0$};
\draw[->] (s1) .. controls (6,1) .. (s2) node[midway,above]{$b, y = 1$}
node[midway,below]{$y := 0$};
\draw[->] (s2) .. controls (6,-1) .. (s1) node[midway,below]{$a, y>0 \wedge x < 1$}
node[midway,above]{$x:=0$};

\draw[->] (-0.7, 0) -> (s0);

\end{tikzpicture}

\caption{A timed automaton which does not preserve qualitative behaviors in sampled semantics. This example is adapted from~\cite{alur94theory}.}
\label{Fig:no-sampling}
\end{figure}
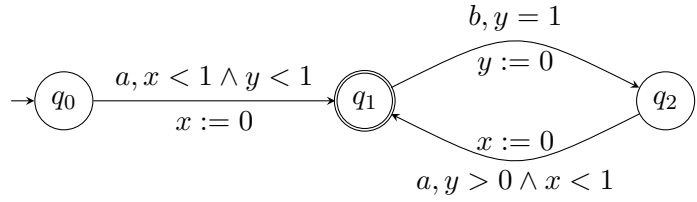

\forget{
\begin{figure}[tbhp]
  \begin{center}

    \[ \xymatrix{
      \ar[r] & *++[o][F=]{l_0} \ar@/^.5cm/[rrrr]^{a,y=1, y:=0} & & & & *++[o][F-]{l_1}
      \ar@/^.5cm/[llll]^{b,1>y \wedge x>1, x:=0} \\
      }
      \]

    \caption{A timed automaton which does not preserve qualitative behaviors in sampled semantics. This example is adapted from~\cite{alur94theory}.}
    \label{Fig:no-sampling-xy}
  \end{center}
\end{figure}
}

\paragraph{\bf Region graph.} We introduce the region equivalence and the standard notion of region graph. Our concept of region equivalence differs from the standard definition in the following technical detail: we consider also the fractional parts of the clocks with the integral part greater than the maximal constant (but we consider only integral parts smaller than or equal to the maximal constant). The important properties of the standard region equivalence (untimed bisimilarity of the equivalent valuations and finite index) are preserved in our definition.

Let for any $\realnum \in \realsnull$, $\intg(\realnum)$ denote the integral part of $\realnum$ and $\fract(\realnum)$ denote the fractional\note{Is this called like this when $\realnum$ is a real?} part of $\realnum$. Let $k$ be an integer constant. For a set of clocks $\clocks$, the relation $\cong_k$ on the set of clock valuations is defined as follows:

\begin{enumerate}[$\bullet$]
\item $\nu \cong_k \nu'$ if and only if all the following conditions hold:

\begin{enumerate}[$-$]
\item for all $x \in \clocks: \intg(\nu(x)) = \intg(\nu'(x))$  or ($\nu(x)>k
  \wedge \nu'(x)>k$),
\item for all $x,y\in \clocks:$
  $\fract(\nu(x)) < \fract(\nu(y))$ if and only if $\fract(\nu'(x)) <  \fract(\nu'(y))$ and   $\fract(\nu(x)) = \fract(\nu(y))$ if and only if $\fract(\nu'(x)) = \fract(\nu'(y))$,
\item for all $x \in \clocks: \fract(\nu(x)) = 0$ if and only if
  $\fract(\nu'(x))= 0$.
\end{enumerate}
\end{enumerate}

\noindent Let $A$ be a timed automaton and $K$ be the maximal constant which occurs in some guard in $A$. For each location $q \in Q$ and two valuations $\nu \cong_K \nu'$ it holds that $(q,\nu)$ is untimed bisimilar to $(q,\nu')$. Also, $\cong_k$ has a finite index for all semantics. We call  equivalence classes of the region equivalence $\cong_K$ \emph{regions} of $A$ and denote them by $D, D', D_1, \dots$. 
For a region $D$ the region $D'$ is the \emph{immediate time successor} if $D' \neq D$, there is $\nu \in D, \realnum\in \realsnull$ such that $\nu+\realnum \in D'$, and for all $\nu \in D, \realnum\in \realsnull$ such that $\nu + \realnum \in D'$ it holds that $\nu+\realnum' \in D\cup D'$ for all $\realnum'\leq \realnum$.

Let $\delta$ be a letter such that $\delta\notin\Sigma$. Given a timed automaton $A = (Q, \actions, \clocks, q_0$, $E$, $F)$, its \emph{region graph} $G = \langle N, \Sigma\cup\{\delta\}, \arrow{} \rangle$ is a labeled directed graph where the set of nodes $N$ contains pairs $\langle q, D \rangle$, where $q$ is a location of $A$ and $D$ is a region of $A$ and $\arrow{} \subseteq N \times \Sigma \cup \{\delta\} \times N$ is a set of labeled edges. Informally, the edges lead to an immediate time successor (labeled by $\delta$) or a discrete successor (labeled by a letter from $\actions$). Formally, $\langle q,D \rangle \arrow{\delta} \langle q, D' \rangle$ if $D'$ is the immediate time successor of $D$ and $\langle q,D\rangle \arrow{a} \langle q', D' \rangle$ if $(l, a, g, Y, l') \in E$, $\nu \models g$ for all $\nu \in D$ and $D' = \{\nu[Y:=0] | \nu \in D\}$.
\note{Say that this definition is standard + reference and that is why we do not show that it is well-defined.}

For a path in the region graph $\sigma = \langle q_1,D_1 \rangle \arrow{w} \langle q_k, D_k \rangle$ we say that a run of the timed automaton in the real or $\epsilon$-sampled semantics (a path in $\densesemantics{A}$ or $\sampledsemantics{A}$, respectively) $\rho = \langle \bar{q_1},\nu_1 \rangle \arrow{w} \langle \bar{q_l},\nu_l \rangle$ is \emph{along this path} if $k = l$ and for all $1\leq i \leq k$, $\langle q_i,D_i \rangle$ is the $i$-th node in $\sigma$, $\langle \bar{q_i},\nu_i \rangle$ is the $i$-th state in $\rho$, $q_i = \bar{q_i}$ and $\nu_i \in D_i$. We denote this by $\rho \models \sigma$.

By $D_\epsilon$, where $\epsilon = 1/k$ for some $k\in\Nat$, we denote the region $D$ restricted to the valuations from the $\epsilon$-sampled semantics. I.e., for all $\nu\in D_\epsilon$, we have that $\nu\in D$ and for all clocks $x$, $\nu(x) = l_x \cdot \epsilon$, where $l_x \in \Nat$.

\subsection{Notation for Clock Differences and Regions}
\label{Subsec:Notation}

We introduce the following notation frequently used in Section~\ref{Sec:Encoding-TA-RC}. For two clocks $b$ and $d$ and a clock valuation $\nu$, we write $\numodels{bd}$ to denote the difference between the fractional parts of the clocks $b,d$ in the valuation $\nu$.
%
%
The distance says how much to the right do we have to move the left clock ($b$ in our case), where the movement to the right wraps at $1$ back to $0$, to reach the right clocks ($d$ in our case). The concept is demonstrated in Figure~\ref{Fig:region-distances}. This figure depicts a valuation of clocks $a,b,c,d$, whose integral values are zero (but they are irrelevant for this definition) and whose fractional parts are set according to the figure ($\nu(a) = 0, \nu(b) = 0.25, \nu(c) = 0.55, \nu(d) = 0.75$). The fractional part of $d$ is greater than that of $b$ and hence to compute $\numodels{bd}$ we simply record how much do we need to move $b$ to the right to reach $d$. This distance is depicted by the (green) dashed arrow above the solid horizontal line. The fractional part of $c$ is greater than that of $b$ and hence to compute $\numodels{cb}$ we need to move $c$ to the right until it reaches $1$, then it wraps (jumps) to $0$, and then we move it further to the right to reach $b$. This distance is depicted by the (red) dashed arrow(s) below the solid horizontal line.

  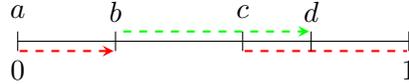
\begin{figure}[htbp]
    \centering
    \begin{tikzpicture}[>=stealth, scale=1.3]

      \path(-2,0) coordinate (0);
      \path(2,0) coordinate (1);
      \path(-2,.1) coordinate (0u);
      \path(-2,-.1) coordinate (0l);
        \path(-1.98,-.1) coordinate (0lr);
      \path(2,.1) coordinate (1u);
      \path(2,-.1) coordinate (1l);
        \path(1.98,-.1) coordinate (1ll);
      \path(-1,.1) coordinate (bu);
        \path(-.92,.1) coordinate (bur);
      \path(-1,-.1) coordinate (bl);
        \path(-1.02,-.1) coordinate (bll);
      \path(-.2,.3) coordinate (xu);
      \path(-.4,0) coordinate (xl);
      \path(.3,.1) coordinate (cu);
      \path(.3,-.1) coordinate (cl);
        \path(.32,-.1) coordinate (clr);
      \path(1,.1) coordinate (du);
        \path(.98,.1) coordinate (dul);
      \path(1,-.1) coordinate (dl);
      \draw(-2,-.3) node [name=zero]{$0$};
      \draw(2,-.3) node [name=one]{$1$};
      \draw(-2,.3) node [name=x]{$a$};
      \draw(-1,.3) node [name=b]{$b$};
      \draw(.3,.3) node [name=c]{$c$};
      \draw(1,.3) node [name=d]{$d$};

      \draw (0)--(1);
      \draw (0u)--(0l);
      \draw (1u)--(1l);
      \draw (bu)--(bl);
      \draw[->,dashed,thick,green] (bur)--(dul);
      \draw[dashed,thick,red] (clr)--(1ll);
      \draw[->,dashed,thick,red] (0lr)--(bll);
      \draw (cu)--(cl);
      \draw (du)--(dl);

    \end{tikzpicture}
    \caption{Illustration of a valuation $\nu$ and the distances between the fractional parts of the clocks. The values of the clocks are $\nu(a) = 0, \nu(b) = 0.25, \nu(c) = 0.55, \nu(d) = 0.75$. The distance between $b$ and $d$ is $\numodels{bd} = 0.5$, the distance between $c$ and $b$ is $\numodels{cb} = 0.7$. Later on, we use this type of diagram only for regions and not for valuations.}
    \label{Fig:region-distances}
  \end{figure}

%
\noindent Formally, for clocks $x,y$ and a clock valuation $\nu$, $\numodels{xy}$ is defined as follows.
    \[ \numodels{xy} = \left\{ \begin{array}{lp{0.3cm}l}
        \fract(\nu(y)) - \fract(\nu(x)) && \mbox{if\ } \fract(\nu(y)) \geq \fract(\nu(x)) \\
        1 - (\fract(\nu(x)) - \fract(\nu(y)))  && \mbox{otherwise}\end{array} \right. \]

We also need to talk about the order of the clocks in a region (an equivalence class of a region equivalence). We say that a region $D$ satisfies an (in)equality $x \bowtie y$ or $x = 0$ (written $\Dmodels{x}{\bowtie}{y}, \Dmodels{x}{=}{0}$) where $\bowtie \in \{<, >, =, \leq, \geq, \neq\}$ if it is true of the fractional parts of x and y in all valuations in the region. Formally, $\Dmodels{x}{\bowtie}{y}$ if for all $\nu\in D$,  $\fract(\nu(x)) \bowtie \fract(\nu(y))$ and $\Dmodels{x}{=}{0}$ if for all $\nu\in D$,  $\fract(\nu(x)) = 0$. Note, that for a given region $D$, either $\fract(\nu(x)) \bowtie \fract(\nu(y))$ holds for all the valuations $\nu\in D$ or it holds for none. Therefore, we adopt the graphical illustration of regions shown in Figure~\ref{Fig:region-example}. Here, a region $D$ is depicted, where $\fract(\nu(a)) = 0$, $\fract(\nu(a)) < \fract(\nu(b)) = \fract(\nu(e)) < \fract(\nu(c)) < \fract(\nu(d))$ for all $\nu\in D$.

  \begin{figure}[htbp]
    \centering
    \begin{tikzpicture}[>=stealth, scale=1.3]

      \path(-2,0) coordinate (0);
      \path(2,0) coordinate (1);
      \path(-2,.1) coordinate (0u);
      \path(-2,-.1) coordinate (0l);
        \path(-1.98,-.05) coordinate (0lr);
      \path(2,.1) coordinate (1u);
      \path(2,-.1) coordinate (1l);
        \path(1.98,-.05) coordinate (1ll);
      \path(-1,.1) coordinate (bu);
        \path(-.92,.05) coordinate (bur);
      \path(-1,-.1) coordinate (bl);
        \path(-1.05,-.05) coordinate (bll);
      \path(-.2,.3) coordinate (xu);
      \path(-.4,0) coordinate (xl);
      \path(.3,.1) coordinate (cu);
      \path(.3,-.1) coordinate (cl);
        \path(.32,-.05) coordinate (clr);
      \path(1,.1) coordinate (du);
        \path(.98,.05) coordinate (dul);
      \path(1,-.1) coordinate (dl);
      \draw(-2,-.3) node [name=zero]{$0$};
      \draw(2,-.3) node [name=one]{$1$};
      \draw(-2,.3) node [name=x]{$a$};
      \draw(-1,.3) node [name=b]{$b,e$};
      \draw(.3,.3) node [name=c]{$c$};
      \draw(1,.3) node [name=d]{$d$};

      \draw (0)--(1);
      \draw (0u)--(0l);
      \draw (1u)--(1l);
      \draw (bu)--(bl);
      \draw (cu)--(cl);
      \draw (du)--(dl);

    \end{tikzpicture}
    \caption{Illustration of the order of the fractional parts of the clocks in a region $D$.}
    \label{Fig:region-example}
  \end{figure}
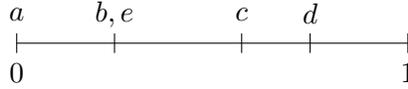

The last concept defined here relates the position of three clocks in a region. For clocks $x,y,z$ and a region $D$, $D \models \overline{xyz}$ tells us that if we start from $x$ and move to the right (and possibly wrap at $1$ back to $0$), we meet $y$ before we meet $z$. Formally, $D \models \overline{xyz}$ if there is a time successor $D'$ of $D$ such that $\Dpmodels{x}{<}{y}$ and $\Dpmodels{y}{<}{z}$. In Figure~\ref{Fig:region-example}, $D \models \overline{bcd}, D
\models \overline{cdb}$, holds, but it is not true that, e.g., $D \models \overline{dcb}$.

\section{Results}
\label{Sec:Results}

We state the main result of this paper -- that our problems are decidable -- and sketch the scheme of a proof of this result.

\begin{thm}
\label{Thm:Decidability}
Given a timed automaton $A$, it is decidable whether there is an $\epsilon = 1/k$ for some $k\in\Nat$ such that
\begin{enumerate}[$\bullet$]
  \item $L_\epsilon(A) = L(A)$ and
  \item $L^\omega_\epsilon(A) = L^\omega(A)$.\qed
\end{enumerate}
\end{thm}



\noindent First, we claim that this theorem is true for timed automata with less than two clocks. It is trivially true for timed automata without clocks ($|\clocks| = 0$). In Section~\ref{Sec:proof} we show that for a timed automaton $A$ with only one clock ($|\clocks| = 1$), $L_{1/2}(A) = L(A)$ and $L_{1/2}^\omega(A) = L^\omega(A)$. We assume that $|\clocks| \geq 2$ in the rest of the paper.

In Section~\ref{Sec:extended-R-automata} we develop a tool of independent interest -- a non-trivial extension of R-automata. These automata contain unbounded counters which can be incremented, reset to zero, copied into each other, and updated by a special type of $\max$ operations. We show that the limitedness problem, i.e., whether there is a bound such that all accepted words can be also accepted by runs along which the counters are smaller than this bound, is decidable for these automata.

The proof of decidability of the sampling problem for timed automata with more than one clock consists of several steps depicted in Figure~\ref{Fig:big-picture}. We start with a given timed automaton $A$. The first step is of a technical character. We transform the timed automaton $A$ into an equivalent timed automaton $A'$ with respect to sampling which never resets more than one clock along each transition. In the second step, we build the region graph $G$ for this timed automaton $A'$. The essential part of the proof is then the third step. Here we transform the region graph $G$ into an extended R-automaton $R$ such that each run in $R$ has a corresponding path in $G$ and vice versa. Moreover, for each run in $R$ and the corresponding path in $G$, there is a relation between the sampling rate which allows for a concrete run along the path and the maximal counter value along the run. The automaton $R$ operates on an extended alphabet -- we have inherited one additional letter $\delta$ for time pass transitions from the region graph. In the last step, we remove the transitions labeled by $\delta$ and build another extended R-automaton $R'$ such that the timed automaton $A'$ can be sampled if and only if $R$ is limited. This step makes use of the fact that the transitions labeled by $\delta$ do not change the counter values, which allows us to use the standard algorithm for removing $\epsilon$-transitions in finite automata.

\begin{figure}[htbp]
\centering
\begin{tikzpicture}[>=stealth]

\draw(0,0) node [circle,double,draw,name=s0]{TA $A$};
\draw[xshift= 6cm](0,0) node [circle,double,draw,name=s1]{TA $A'$};
\draw[xshift= 12cm,yshift=0cm](0,0) node [circle,double,draw,name=s2]{RG $G$};
\draw[xshift= 6cm,yshift=-3cm](0,0) node [circle,double,draw,name=s3]{ERA $R$};
\draw[xshift= 12cm,yshift=-3cm](0,0) node [circle,double,draw,name=s4]{ERA $R'$};

\draw[->] (s0) -- (s1) node[midway,above]{remove multiple resets}
                       node[midway,below]{{\it sampling equivalent}};
\draw[->] (s1) -- (s2) node[midway,above]{standard construction};
\draw[->] (s2) -- (s3) node[midway,above,sloped]{our new construction};
\draw[->] (s3) -- (s4) node[midway,above,sloped]{remove $\delta$ transitions}
                       node[midway,below,sloped]{{\it preserves counter values}};

\draw[dashed] (6cm,-1.5cm) ellipse (1.1cm and 2.5cm);
\draw (0.2cm,-1.8cm) node[right] {{\it Relation between counter}};
\draw (0.2cm,-2.3cm) node[right] {{\it values and $\epsilon$ along all runs}};

\end{tikzpicture}
   \caption{An overview of the proof structure. The abbreviations $TA$, $RG$, and $ERA$ stand for Timed Automaton, Region Graph, and Extended R-Automaton, respectively.}
   \label{Fig:big-picture}
\end{figure}
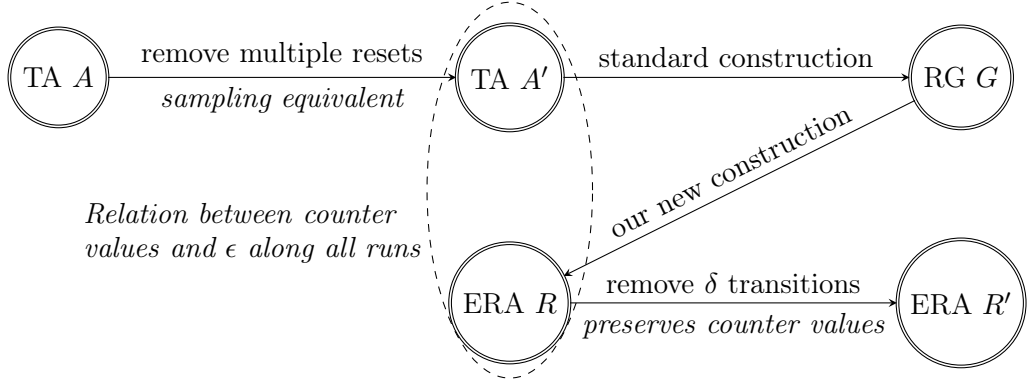

The first and the last step are rather straightforward and we show them in Section~\ref{Sec:proof}. The new model of extended R-automata is presented in Section~\ref{Sec:extended-R-automata}. Section~\ref{Sec:Reduction-maxplus-R} shows how to reduce the limitedness problem for extended R-automata to the limitedness problem of R-automata, which was shown decidable in~\cite{aky08r-automata}. Finally, the main reduction step, the translation of a region graph (induced by $A'$) into an extended R-automaton $R$ and the proof of relation between $A'$ and $R$, together with an informal overview is shown in Section~\ref{Sec:Encoding-TA-RC}.

\section{Extended R-automata}
\label{Sec:extended-R-automata}

In this section, we present an extension of R-automata. R-automata are finite state machines with counters which can be updated by the following instructions: no update, increment and reset to zero ($0,1,r$, respectively). We extend the set of instructions by a copy of one counter value into another counter and taking a maximum of the counters and sums of pairs of counters under specific conditions. For this extension, we show that the limitedness problem is decidable by a reduction to the universality problem of R-automata, shown decidable in~\cite{aky08r-automata}. 

\subsection{Extensions of R-automata}

Before we define syntax and semantics of extended R-automata, we give some informal introduction. The first extension is adding the ability to copy the value of one counter into another counter. The instruction set is extended by instructions $*j$, where $j$ is a counter name and applying this instruction to a counter $i$ results in the counter $i$ having the value of the counter $j$.

The other extension we need in order to reduce our problems for timed automata to limitedness of counter automata (taking maxima of counters and counter sums) is rather semantical than syntactical. The only syntactical change is that the reset instruction is equipped with a subset of counters, i.e., if $n$ is the number of counters, reset instructions are $r(A), A \subseteq\{1, \dots, n\}$. The semantics maintains three values for each counter ($P, M, N$) and a preorder $\lesssim$ on the counters. This rather nonstandard terminology -- a counter containing three values -- makes the definitions in this section and proofs in Section~\ref{Sec:Reduction-maxplus-R} simpler. One can see this as if for a counter $i$ we now have three new counters $P_i$, $M_i$, and $N_i$.

The values $N_i$ behave in the same way as for R-automata with copying. The preorder tells us how to apply the $\max$ operation to the values $P$ and $M$. These values of a counter $j$ are always greater than these values of a counter $i$ such that $i \lnsim j$. More concretely, if $i \lnsim j$ then $M_j \geq M_i + 1$ and if $k,l \lnsim j$ then $P_j \geq P_k + P_l$. The way in which we update the preorder $\lesssim$ along the transitions ensures that, informally, for all counters $i$, the values $P_i$ and $M_i$ cannot grow unbounded along a run where $N_i$ is bounded.

\paragraph{\bf Syntax.} Let for a given number $n$ of counters, $\Eff = \{0,1\} \cup \{r(A) | A \subseteq \{1, \dots, n\}\} \cup \{*m | 1\leq m \leq n\}$ be the set of instructions on a counter.
An \emph{extended R-automaton} with $n$ counters is a $5$-tuple $R = (S, \Sigma, \Delta, s_0, F)$ where
\begin{enumerate}[$\bullet$]
\item $S$ is a finite set of states,
\item $\Sigma$ is a finite alphabet,
\item $\Delta \subseteq S\times \Sigma \times \Eff^n \times S$ is a transition relation,
\item $s_0 \in S$ is an initial state, and
\item $F \subseteq S$ is a set of final states.
\end{enumerate}

Transitions are labeled (together with a letter) by an effect on the counters. The symbol $0$ corresponds to leaving the counter value unchanged, the symbol $1$ represents an increment, the symbol $r(A)$ represents a reset (the function of $A$ will be explained later), and a symbol $*j$ means that the value of this counter is set to the value of the counter $j$. The instructions $0,1,$ and $r(A)$ take place first and after that the values are copied. An automaton which does not contain any copy instruction and all resets contain an empty set is called an R-automaton (effects contain only $0,1,r(\emptyset)$). We skip the subset of counters $A$ and write $r$ instead of $r(A)$ when the set does not play any role (e.g., in the whole of Section~\ref{Sec:Reduction-copy-R}).

We use $t, t', t_1, \dots$ to denote elements of $\Eff^n$ which we call \emph{effects}.
By $\pi_i(t)$ we denote the $i$-th projection of $t$. Without loss of generality, we assume that the value of a counter is never directly copied into itself ($\pi_i(t) \neq *i$). A \emph{path} is a sequence of transitions $(s_1, a_1, t_1, s_2)$,$(s_2, a_2, t_2, s_3),\dots$, $(s_m,a_m,t_m, s_{m+1})$, such that $\forall 1\leq i \leq m. (s_i,a_i,t_i,s_{i+1})\in\Delta$. We use $s_i$ to refer to the $i$-th state of the path. An example of an extended R-automaton is given in Figure~\ref{Fig:automaton}.

\begin{figure}[htbp]
\centering
\begin{tikzpicture}[>=stealth]

\draw(0,0) node [circle,draw,name=s0]{$s_0$};
\draw[xshift= 3cm](0,0) node [circle,double,draw,name=s1]{$s_1$};
\draw[xshift= 3cm,yshift=-2cm](0,0) node [circle,double,draw,name=s2]{$s_2$};

\draw[->] (s0) -- (s1) node[midway,above]{$a,(1,0)$};
\draw[->] (s1) -- (s2) node[midway,right]{$b,(r(\emptyset),*1)$};
\draw[->] (s0) -- (s2) node[midway,below,sloped]{$a,(0,1)$};

\draw[->] (s1) .. controls (2.5,1) and (3.5,1) .. (s1) node[midway,above]{$b,(0,1)$};
\draw[->] (s2) .. controls (4,-1.5) and (4,-2.5) .. (s2) node[midway,right]{$a,(0,r(\{1\}))$};

\draw[->] (-0.7, 0) -> (s0);

\end{tikzpicture}
   \caption{An R-automaton with two counters.}
   \label{Fig:automaton}
\end{figure}
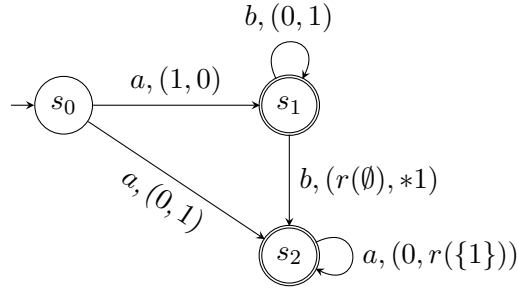

\paragraph{\bf Unparameterized semantics.} We define an operation $\plus$ on the counter values: for any $k\in \Nat$, $k\plus 0 = k$, $k\plus 1 = k+1$, and $k\plus r = 0$. We extend this operation to $n$-tuples and copy instructions as follows. For a $t\in\Eff^n$, let $\hat{t}$ be an effect with all copy instructions replaced by $0$, i.e., $\pi_i(\hat{t}) = \pi_i(t)$ if $\pi_i(t)\in \{0,1,r(A)\}$ and $\pi_i(\hat{t}) = 0$ otherwise. For a $t\in\Eff^n$ and $(c_1, \dots, c_n) \in \Nat^n$, $(c_1, \dots, c_n) \plus t = (c_1', \dots, c_n')$, where $c_i' = c_j \plus \pi_j(\hat{t})$ if $\pi_i(t) = *j$ for some $j$ and $c_i' = c_i \plus \pi_i(t)$ otherwise. For example, $(1,5,7) \plus (1, *1, *2) = (2, 2, 5)$ -- first we increment the first counter and then we copy the values of the first and the second counter into the second and the third counter, respectively.

The operational semantics of an extended R-automaton $R = (S, \Sigma, \Delta, s_0, F)$ is given by an LTS $\unpsemantics{R} = ( \SS, \Sigma, T, \hat{s_0} )$, where the set of states $\SS$ contains triples $\langle s, \bar{C}, \lesssim \rangle$, $s\in S, \bar{C} \in \Nat^n \times \Nat^n \times \Nat^n$, $\lesssim$ is a preorder on $\{1, \dots, n\}$, with the initial state $\hat{s_0} = \langle s_0, \bar{C_0}, \emptyset \rangle$, where $\bar{C_0} = (0^n, 0^n, 0^n)$.
For a $\bar{C} \in \Nat^n \times \Nat^n \times \Nat^n$, we denote the first projection by $\bar{P}$, the second projection by $\bar{M}$, and the third projection by $\bar{N}$. I.e., $\bar{P}, \bar{M}, \bar{N} \in \Nat^n$ and $\bar{C} = (\bar{P}, \bar{M}, \bar{N})$.
For $1\leq i \leq n$, we denote by $P_i$, $M_i$, or $N_i$ the $i$-th projection of $\bar{P}$, $\bar{M}$, or $\bar{N}$, respectively.
The role of the preorder $\lesssim$ and of the counter valuation is informally explained below the formal definition of the transition relation. We introduce a shorthand $i \simeq j$ for $(i \lesssim j \wedge j \lesssim i)$ and $i \lnsim j$ for $(i \lesssim j \wedge \neg(j \lesssim i))$.

The transition relation is defined as follows: $(\langle s, \bar{C}, \lesssim \rangle, a$, $\langle s'$, $\bar{C'}, \lesssim' \rangle) \in T$ if and only if $(s, a, t, s') \in \Delta$ and $\bar{C'}, \lesssim'$ are constructed by the following three steps (executed in this order):

\begin{enumerate}[(1)]
  \item \label{Sem:effects} $\bar{P}' = \bar{P} \plus t, \bar{M}' = \bar{M} \plus t$, and $\bar{N}' = \bar{N} \plus t$ 
  \item \label{Sem:preorder} The preorder $\lesssim'$ is constructed in two steps. First, $i \lesssim' j$ if and only if either:
  \begin{enumerate}[(a)]
    \item \label{Pre:plus} $i \lesssim j$ and $\pi_i(t) \in \{0,1\}, \pi_j(t) \in \{0,1\}$ and it is not true that $j \lesssim i$, $\pi_i(t) = 1$, and $\pi_j(t) = 0$, or
    \item \label{Pre:reset} $\pi_i(t) = r(\{j\} \cup A)$ and $N_j' > 0$, or 
    \item \label{Pre:copy} $\pi_i(t) = *j$ or $\pi_j(t) = *i$.
  \end{enumerate}
  Secondly, add the transitive and reflexive closure to $\lesssim'$.
  \item \label{Sem:max} Repeat the following until a fixed point is reached: if $i \lnsim j$ then set $M_j' = \max\{M_j',$ $M_i' + 1\}$ and if $k,l \lnsim j$ then $P_j' = \max\{P_j', P_k' + P_l'\}$.
\end{enumerate}

\noindent We shall call the states of the LTS \emph{configurations}.
We write $\langle s, \bar{C}, \lesssim \rangle \arrow{a} \langle s', \bar{C'}, \lesssim' \rangle$ if $(\langle s, \bar{C}, \lesssim \rangle, a, \langle s', \bar{C'}, \lesssim \rangle) \in T$. We extend this notation to words, $\langle s, \bar{C}, \lesssim \rangle \arrow{w} \langle s', \bar{C'}, \lesssim' \rangle$, where $w\in\Sigma^+$.

Note that the values $\bar{N}$ of the counters 
are updated only by the instructions $0,1,r(A)$ and $*j$ (Step~\ref{Sem:effects}). The values $\bar{M}$ and $\bar{P}$ of the counters are updated by these effects as well (Step~\ref{Sem:effects}), but they can also be increased by the $\max$ operation (Step~\ref{Sem:max}). Namely, $N_i > 0$ implies that $M_i > 0$ and $P_i > 0$. Clearly, there is always a fixed point reached after at most $n$ iterations of Step~\ref{Sem:max}.

The preorder $\lesssim$ in a reachable state $\langle s, \bar{C}, \lesssim \rangle$ relates counters $i, j$ only if the values $M_i, P_i$ are smaller than or equal to $M_j, P_j$ ($i \lesssim j$ implies $M_i \leq M_j, P_i \leq P_j$). Especially, $i \simeq j$ implies $M_i = M_j, P_i = P_j$. This is satisfied in the initial state (trivially) and preserved by updates in Step~\ref{Sem:preorder}. There, the effects influence the preorder $\lesssim$ in the following way: an equality is broken if one counter is incremented and the other one is left unchanged (Step~\ref{Pre:plus}), a reset removes the counter from the preorder and puts it below non-zero counters indicated in the reset (Step~\ref{Pre:reset}), and a copy instruction sets the counter equal to the counter whose value it copied (Step~\ref{Pre:copy}). In other cases, the relation is preserved (Step~\ref{Pre:plus}). An example of the effect of Step~\ref{Sem:preorder} on a preorder is in Figure~\ref{Fig:preorder-example}.

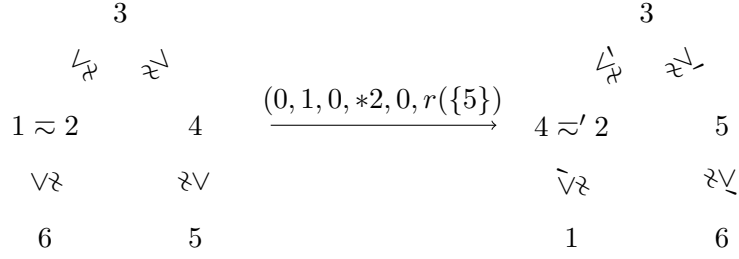
\begin{figure}[htbp]
\centering
\begin{tikzpicture}

\draw(0,0) node [name=s6]{$6$};
\draw[xshift= 0cm,yshift=0.75cm](0,0) node [rotate=90]{$\lnsim$};
\draw[xshift= 2cm](0,0) node [name=s5]{$5$};
\draw[xshift= 2cm,yshift=0.75cm](0,0) node [rotate=-90]{$\gnsim$};
\draw[xshift= 0cm,yshift=1.5cm](0,0) node [name=s12]{$1\eqsim 2$};
\draw[xshift= 0.5cm,yshift=2.3cm](0,0) node [rotate=45]{$\lnsim$};
\draw[xshift= 2cm,yshift=1.5cm](0,0) node [name=s4]{$4$};
\draw[xshift= 1.5cm,yshift=2.3cm](0,0) node [rotate=-45]{$\gnsim$};
\draw[xshift= 1cm,yshift=3cm](0,0) node [name=s3]{$3$};

\draw[->] (3cm,1.5cm) -- (6cm,1.5cm) node[midway,above]{$(0,1,0,*2,0,r(\{5\})$};

\draw(7,0) node [name=sp6]{$1$};
\draw[xshift= 7cm,yshift=0.75cm](0,0) node [rotate=90]{$\lnsim'$};
\draw[xshift= 9cm](0,0) node [name=sp5]{$6$};
\draw[xshift= 9cm,yshift=0.75cm](0,0) node [rotate=-90]{$\gnsim'$};
\draw[xshift= 7cm,yshift=1.5cm](0,0) node [name=sp12]{$4\eqsim' 2$};
\draw[xshift= 7.5cm,yshift=2.3cm](0,0) node [rotate=45]{$\lnsim'$};
\draw[xshift= 9cm,yshift=1.5cm](0,0) node [name=sp4]{$5$};
\draw[xshift= 8.5cm,yshift=2.3cm](0,0) node [rotate=-45]{$\gnsim'$};
\draw[xshift= 8cm,yshift=3cm](0,0) node [name=sp3]{$3$};


\end{tikzpicture}
  \caption{An example of updates of $\lesssim$ after applying the effect $(0,1,0,*2,0,r(\{5\}))$. The diagram on the left side depicts $\lesssim$ and the diagram on the right side depicts $\lesssim'$. Step~\ref{Pre:plus} sets $1\lesssim' 3, 2\lesssim' 3, 5\lesssim' 3, 1\lesssim' 2$. It does not set $2\lesssim' 1$, because the counter $2$ was incremented while the counter $1$ was left unchanged. Step~\ref{Pre:reset} sets $6\lesssim' 5$ (we assume that $N'_5 > 0$). Step~\ref{Pre:copy} sets $4\lesssim' 2, 2\lesssim' 4$. The transitive and reflexive closure completes $\lesssim'$ to a preorder.}
  \label{Fig:preorder-example}
\end{figure}

\noindent Another view on the preorder is what sequence of effects results in $i\lnsim j$. This can happen only in the following three ways. First, when $i$ is reset with $j$ in the set, i.e., by $r(\{j\} \cup A)$, and $M_j > 0$. Second, $i$ is copied to $j$ or $j$ is copied to $i$ and then $j$ is incremented by $1$ while $i$ stays unchanged (the instruction is $0$). 
Third, the relation $i\lnsim j$ can also be a result of the transitive closure. If already $i\lnsim j$ holds then it can be broken only by a reset or a copy of one of these two counters.

The preorder $\lesssim$ influences only the values $\bar{P}$ and $\bar{M}$. If we skip Step~\ref{Sem:preorder} in the semantics (which would result in $\lesssim$ to be empty  in all the reachable states) then $\bar{P} = \bar{M}  = \bar{N}$ in all the reachable states. Also, changes of the values $\bar{N}$ along a transition depend only on the effect and not on $\lesssim$ in the starting state.

We could also view our extension as R-automata which can perform $\max$ operations on the counters along the transitions. The motivation for introducing the preorder $\lesssim$ instead of allowing explicit $\max$ operations as instructions on the transitions is to restrict the usage of $\max$ operations so that Lemma~\ref{Lemma:counter-maxs} and Lemma~\ref{Lemma:counter-sums} hold. Unrestricted usage of $\max$ operation is equivalent to alternation. Limitedness has been shown decidable for alternating cost tree automata in~\cite{CSL08:colcombet-loeding-depth-mu-calculus}, but resets have to follow a hierarchical (parity-like) discipline in these automata and copying in not allowed.

Paths in an LTS are called \emph{runs} to distinguish them from paths in the underlying extended R-automaton. Observe that the LTS contains infinitely many states, but the counter values do not influence the computations, since they are not tested anywhere. In fact, for any extended R-automaton $R$, $\unpsemantics{R}$ is bisimilar to $R$ considered as a finite automaton (without counters and effects).

\paragraph{\bf Parameterized Semantics.} Next, we define $B$-semantics of extended R-auto\-ma\-ta. The parameter $B$ is a bound on the counter values $\bar{N}$ which can occur along any run. For a given $B\in\Nat$, let $\SS_B$ be the set of configurations restricted to the configurations which do not contain a counter whose $\bar{N}$ values exceed $B$, i.e., $\SS_B = \{ \langle s, \bar{C}, \lesssim \rangle~|$ $\langle s, \bar{C}, \lesssim \rangle \in \SS \wedge \bar{C}= (\bar{P}, \bar{M}, \bar{N}) \wedge \forall 1\leq i \leq n. N_i \leq B\}$. For an extended R-automaton $R$, the \emph{$B$-semantics} of $R$, denoted by $\dsemantics{R}$, is $\unpsemantics{R}$ restricted to $\SS_B$. We write $\langle s, \bar{C}, \lesssim \rangle \arrow{a}_B \langle s', \bar{C'}, \lesssim' \rangle$ to denote the transition relation of $\dsemantics{R}$. We extend this notation to words, $\langle s, \bar{C}, \lesssim \rangle \arrow{w}_B \langle s', \bar{C'}, \lesssim \rangle$, where $w\in\Sigma^+$.

\paragraph{\bf Language.} The (unparameterized or $B$-) language of an extended R-automaton is the set of words which can be read along the runs in the corresponding LTS ending in an accepting state (a configuration whose first component is an accepting state). Formally, for a run $\rho$ in $\unpsemantics{R}$, let $l(\rho)$ denote the concatenation of the labels along this run. A run $\rho = \langle s_0, \bar{C_0}, \emptyset \rangle \arrow{}^* \langle s, \bar{C}, \lesssim \rangle$ is accepting if $s\in F$. The \emph{unparameterized language} accepted by an extended R-automaton $R$ is $L(R) = \{ l(\rho) | \rho$ is an accepting run in $\unpsemantics{R} \}$. For a given $B\in\Nat$, the \emph{$B$-language} accepted by an extended R-automaton $R$ is $L_B(A) = \{ l(\rho) | \rho$ is an accepting run in $\dsemantics{R} \}$. The unparameterized language of the extended R-automaton from Figure~\ref{Fig:automaton} is $ab^*a^*$. The $2$-language of this automaton is $a(\epsilon+b+bb+bbb)a^*$. We also in the standard way define the language of infinite words for R-automata with B\"{u}chi acceptance conditions, denoted by $L^\omega(R), L_B^\omega(R)$.

\paragraph{\bf Limitedness/Universality.} The language of an extended R-automaton $R$ is \emph{limited} or \emph{universal} if there is a natural number $B$ such that $L_B(R) = L(R)$ or $L_B(R) = \Sigma ^*$, respectively. The definition of these problems for $\omega$-languages is analogous.
We show in Lemma~\ref{Lemma:RC-automata} that it is decidable whether a given extended R-automaton is limited or universal and in Lemma~\ref{Lemma:counter-maxs} and Lemma~\ref{Lemma:counter-sums} that this concept would not change even if we limit the $\bar{P}$ or $\bar{M}$ values in the definition of $B$-semantics.

We could split an extended R-automaton into three different automata which would maintain only one of the values $\bar{P}, \bar{M}, \bar{N}$. Later on, in the reduction from timed automata to these automata, we use only $\bar{P}$ values. The presentation which we chose (all values together in one automaton) simplifies the notation for the proofs of Lemma~\ref{Lemma:counter-maxs} and Lemma~\ref{Lemma:counter-sums}.

\subsection{Limitedness of Extended R-automata -- Copy Operations}
\label{Sec:Reduction-copy-R}

First, we show that the limitedness problem for extended R-automata is decidable. In this section, we deal only with the $N$ values of extended R-automata. We ignore the preorder $\lesssim$ (as it is not needed for calculating the $N$ values) and when we say that a counter $i$ has a value $k$ then we mean that $N_i = k$. We also write only $r$ instead of $r(A)$. The decidability proof reduces the limitedness problem for extended R-automata to the limitedness problem of R-automata. It has been shown in~\cite{aky08r-automata} that the universality problem of R-automata is decidable, but it is easy to see that this procedure can be used also to decide the limitedness problem. We create a disjoint union of the R-automaton in question and its complement (where the automaton is considered without effects, as a standard finite automaton). We add effects $(0, \dots, 0)$ on all transitions of the complement. This automaton is universal if and only if the original R-automaton is limited.

\begin{lemma}
\label{Lemma:RC-automata}
  For a given extended R-automaton $R$, the questions whether there is $B\in\Nat$ such that $L_B(R) = L(R)$ (and $L_B^\omega(R) = L^\omega(R)$) is decidable.
  \qed
\end{lemma}

The rest of this subsection proves this lemma. In order to avoid unnecessary technical complications in the main part of the proof, we restrict ourselves to extended R-automata with at most one copy instruction in each effect. We show how to extend the proof to the general model at the end of this subsection. We reduce the universality problem for extended R-automata to the universality problem of R-automata, for which this problem has been shown decidable in~\cite{aky08r-automata}.

\paragraph{\bf Construction.}
As the first step, we equip each R-automaton with a variable called \emph{parent pointer} for each counter and with the ability to swap the values of the counters. The parent pointers range over $\{\Null\} \cup \{1, \dots, n\}$, where $n$ is the number of the counters. We shall use them to capture (a part of) the history of copying. We observe that for each R-automaton one can encode the value swapping and the parent pointers into the states. To express properties of this encoding more formally, let us assume that the transitions in the semantics LTS are labeled also by the counter values (in the order encoded by the automaton) and the parent pointers. For each R-automaton $\hat{R}$ with parent pointers and value swapping, we can build an R-automaton $\bar{R}$ with $|S|\cdot n!\cdot 2^{n}$ states bisimilar to $\hat{R}$, where $|S|$ is the number of the states of $\hat{R}$. Moreover, any number of value swaps and parent pointer operations can be encoded along each transition of $\bar{R}$ together with standard updates (increments, resets). $\bar{R}$ can also branch upon the values of the parent pointers.

\revision{PARENT POINTER MOTIVATION}

Before presenting the construction, we give some informal motivation for using parent pointers and counter swapping. When an automaton copies a value of a counter $i$ to a counter $j$ then, from this time point on, the values in these two counters develop independently. Any of them might eventually exceed an imposed bound. The simulating automaton has, however, only one copy of this value stored in the counter $i$. Therefore, the best the simulating automaton can do is to use this value to track the evolution of one of the two values from the original automaton. For the other value, we start simulating its evolution from $0$ (which is easily done by a reset), hoping that the loss of the value accumulated in the counter $i$ can be bounded in some way.

Let us look a bit closer on what do we mean by evolution of a value (formalized as a value trace in Definition~\ref{Def:value-trace} below). A value contained in a counter $i$ after $t$ computation steps is \emph{alive} after $k$ additional steps of computation (i.e., at the time point $t+k$) if there is a counter whose value at the time point $t+k$ was obtained from the original value (i.e., the value contained in the counter $i$ at the time point $t$) by incrementing and copying ($0, 1, *k$ operations). Each sequence of these operations which witnesses that a value is alive constitutes an evolution of this value. A value dies if all of its copies are reset or overwritten by a copy of some other value.

The simulating automaton has to choose which of the two counters does it want to simulate with the original value accumulated in the counter $i$. We want the automaton to choose the counter whose value stays alive longer. The reason is as follows. There has to be an evolution which witnesses this property. This evolution occupies at least one counter during the whole lifetime of this value. Because the other value lives shorter, it has strictly fewer counters for copying itself. This gives us an inductive argument resulting in an upper bound on the number of simulation resets, i.e., resets introduced to simulate copy operations, along each value evolution (being the number of the counters).

Technically, the automaton chooses a counter non-deterministically (by possibly swapping the values) and it uses parent pointers to verify the correctness of all choices. After each choice, it updates the parent pointers so that a pointer pointing from a counter $i$ to a counter $j$ expresses the guess that the value which is currently in the counter $j$ will live longer than the value in the counter $i$. We are interested only in relations between values which have the same origin (one value was created as a copy of another). Therefore, it is enough to have only one parent pointer for each counter. One can then detect from the parent pointers and an effect whether applying this effect would violate the guesses.

Figure~\ref{fig:parent-pointers-example} depicts a sample run of an extended R-automaton with three counters initialized with zeros. The solid (blue) line denotes an evolution of the initial value of counter $3$ (its value trace). Other lines denote alternative evolutions of the same value, but they are all shorter than the solid (blue) one. Crosses at the ends of the lines show the points where the alternative value dies. The arrows depict parent pointers along a correct run of the simulating automaton. They always connect the traces with the same splitting point and they point from the shorter to the longer one. As an illustration of how do parent pointers serve for detecting wrong guesses, imagine that the parent pointer between counters $2$ and $3$ in the third step (the first one with the effect $(1, *3, 0)$) has been set the other way round, i.e., pointing from the counter $2$ to the counter $3$. At the fifth step (the first one with the effect $(1,1,r)$), the automaton knows directly from the effect that the value in the counter $3$ dies and the value in counter $2$ is still alive. This is not consistent with the parent pointer and the automaton would enter an error state.

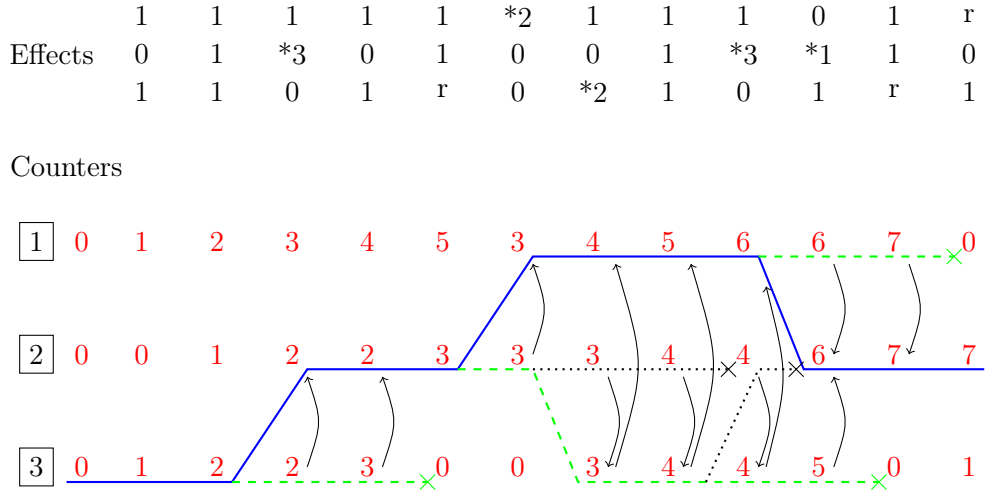
\begin{figure}
\begin{tikzpicture}

\foreach \x/ \n in {0/1,1/1,2/1,3/1,4/1,5/*2,6/1,7/1,8/1,9/0,10/1,11/r}{
\node at (\x,0) {\n};
}

\foreach \x/ \n in {0/0,1/1,2/*3,3/0,4/1,5/0,6/0,7/1,8/*3,9/*1,10/1,11/0}{
\node at (\x,-0.5) {\n};
}

\foreach \x/ \n in {0/1,1/1,2/0,3/1,4/r,5/0,6/*2,7/1,8/0,9/1,10/r,11/1}{
\node at (\x,-1) {\n};
}

\node[anchor=east] at (-0.5,-0.5){Effects};

\node at (-1,-2) {Counters};

\node[draw] at (-1.4,-3) {1};
\node[draw] at (-1.4,-4.5) {2};
\node[draw] at (-1.4,-6) {3};

\foreach \x/ \n in {-0.8/0,0/1,1/2,2/3,3/4,4/5,5/3,6/4,7/5,8/6,9/6,10/7,11/0}{
\node[text=red] at (\x,-3) {\n};
}

\foreach \x/ \n in {-0.8/0,0/0,1/1,2/2,3/2,4/3,5/3,6/3,7/4,8/4,9/6,10/7,11/7}{
\node[text=red] at (\x,-4.5) {\n};
}

\foreach \x/ \n in {-0.8/0,0/1,1/2,2/2,3/3,4/0,5/0,6/3,7/4,8/4,9/5,10/0,11/1}{
\node[text=red] at (\x,-6) {\n};
}

\draw[blue,thick] (-1,-6.2) -- (1.2,-6.2) -- (2.2,-4.7) -- (4.2,-4.7) -- (5.2,-3.2) -- (8.2,-3.2) -- (8.8,-4.7) -- (11.2,-4.7);

\draw[green,thick,dashed] (1.2,-6.2) -- (3.8,-6.2);
\draw[green] (3.8,-6.2)--+(-0.1,-0.1) +(0,0)--+(0.1,0.1) +(0,0)--+(-0.1,0.1) +(0,0)--+(0.1,-0.1);

\draw[green,thick,dashed] (4.2,-4.7) -- (5.2,-4.7) -- (5.8,-6.2) -- (9.8,-6.2);
\draw[green] (9.8,-6.2)--+(-0.1,-0.1) +(0,0)--+(0.1,0.1) +(0,0)--+(-0.1,0.1) +(0,0)--+(0.1,-0.1);

\draw[green,thick,dashed] (8.2,-3.2) -- (10.8,-3.2);
\draw[green] (10.8,-3.2)--+(-0.1,-0.1) +(0,0)--+(0.1,0.1) +(0,0)--+(-0.1,0.1) +(0,0)--+(0.1,-0.1);

\draw[thick,dotted] (5.2,-4.7) -- (7.8,-4.7);
\draw (7.8,-4.7)--+(-0.1,-0.1) +(0,0)--+(0.1,0.1) +(0,0)--+(-0.1,0.1) +(0,0)--+(0.1,-0.1);

\draw[thick,dotted] (7.5,-6.2) -- (8.2,-4.7) -- (8.7,-4.7);
\draw (8.7,-4.7)--+(-0.1,-0.1) +(0,0)--+(0.1,0.1) +(0,0)--+(-0.1,0.1) +(0,0)--+(0.1,-0.1);

\draw[->] (2.2,-6.0) .. controls (2.4,-5.4) .. (2.2,-4.8);
\draw[->] (3.2,-6.0) .. controls (3.4,-5.4) .. (3.2,-4.8);

\draw[->] (5.2,-4.5) .. controls (5.4,-3.9) .. (5.2,-3.3);

\draw[->] (6.2,-4.8) .. controls (6.4,-5.4) .. (6.2,-6);
\draw[->] (6.3,-6) .. controls (6.7,-4.65) .. (6.3,-3.3);

\draw[->] (7.2,-4.8) .. controls (7.4,-5.4) .. (7.2,-6);
\draw[->] (7.3,-6) .. controls (7.7,-4.65) .. (7.3,-3.3);

\draw[->] (8.2,-4.8) .. controls (8.4,-5.4) .. (8.2,-6);
\draw[->] (8.3,-6) .. controls (8.6,-4.8) .. (8.3,-3.6);

\draw[->] (9.2,-6.0) .. controls (9.4,-5.4) .. (9.2,-4.8);

\draw[->] (9.2,-3.3) .. controls (9.4,-3.9) .. (9.2,-4.5);

\draw[->] (10.2,-3.3) .. controls (10.4,-3.9) .. (10.2,-4.5);

\end{tikzpicture}
\caption{An example run of an extended R-automaton with an illustration of value evolution and parent pointers. The twelve effects above are applied to the three counters below in twelve consecutive steps. The solid (blue), dashed (green) and dotted (black) lines depict value evolutions ending with a cross. Arrows show parent pointers in a correct simulating run.}
\label{fig:parent-pointers-example}
\end{figure}

The simulating automaton uses the counter value for the longer value trace (by possibly swapping the counter values) and resets the other counter to $0$. In our example, a value trace splits in two with each copy operation. The value trace which keeps the style (color) is simulated by the counter value, while the one denoted by a different style (color) resets the counter value. The key observation for the simulation correctness is that when the value is reset twice (in our example with three counters) in copy simulations then it cannot be copied to another counter, because it would violate some parent pointers. Therefore, it cannot be reset in another copy simulation anymore. This is the case for the dotted (black) value traces.

\revision{END OF PARENT POINTER MOTIVATION}

Now we can present the reduction by constructing an R-automaton $\hat{R}$ which uses counter value swapping and the parent pointers for each extended R-automaton $R$ such that $\hat{R}$ is limited if and only if $R$ is limited. $\hat{R}$ has all the states of $R$ together with an error sink and it has the same initial state $s_0$ and the same set of accepting states as $R$. The error sink is a non-accepting state with no outgoing transitions except for self-loops labeled by $\Sigma$ and effects $(0, \dots, 0)$ which do not swap any counter values and do not manipulate the parent pointers. The automaton starts in the initial state with all parent pointers set to $\Null$. To define the transitions of $\hat{R}$, we need to encode copying by resets, value swapping and updates of parent pointers. To do this, we replace each copy by a reset, possibly with some (non-deterministic) value swapping and bookkeeping of the parent pointers. 

For each transition of $R$ we either construct simulating transitions or a transition going to the error sink. Let us denote the simulated transition of $R$ by $s \arrow{a,t} s'$, where $t = (e_1, \dots, e_n)$. If there are counters $k,l$ such that $e_k \in \{0,1\}$, $e_l \notin \{0,1\}$, and the parent pointer of $k$ points to (is set to) $l$ then we create a transition going to the error sink. Otherwise, we build simulating transitions $s \arrow{a, t', sp} s'$ in $\hat{R}$ labeled by an effect $t' = (e_1', \dots, e_n')$, which might also swap some counter values and manipulate the parent pointers (denoted by $sp$).

If $t$ does not contain any copy instruction then there is one simulating transition with $t' = t$ and for all $i$ such that $e_i = r$, we set $i$'s parent pointer to $\Null$. No counter values are swapped.

If $t$ contains a copy instruction $e_i = *j$ then we create two simulating transitions. Each of them has the same effect $t' = (e_1', \dots, e_n')$, where $e'_k = e_k$ if $k\neq i$ and $e'_i = r$. These two transitions give the simulating automaton a non-deterministic choice between the counters $i$ and $j$. The first transition corresponds to the choice of $j$. Along this transition, we perform the effect and set $i$'s parent pointer to $j$. No counter values are swapped. Along the other transition (corresponding to the choice of $i$), we perform the effect, swap the values of the counters $i$ and $j$, we copy the value of $j$'s parent pointer into $i$'s parent pointer, we change the value of all parent pointers with value $j$ to $i$,  and finally we set $j$'s parent pointer to $i$.
Both transitions also set the $k$'s parent pointer to $\Null$ for all $k$ such that $e_k = r$. An example of the construction of simulating transitions for a transition with an effect containing a copy instruction is depicted in Figure~\ref{Fig:copy-encoding}.

\begin{figure}[htbp]
\centering
\begin{tikzpicture}[scale=1.2]

\draw(-.1,0) node [name=pre]{$s, (2,5,9,1,8)$};
\draw[xshift= 5cm](0.1,0) node [name=post]{$s', (3,5,5,0,9)$};
\draw[->] (1cm,0cm) -- (4cm,0cm) node[midway,above]{$a, (1,0,*2,r,1)$};

\draw[yshift= -2.25cm](-.1,0) node [name=pre]{$s, (2,5,9,1,8)$};
\draw[xshift= 5cm,yshift= -1.5cm](0.1,0) node [name=post]{$s', (3,5,0,0,9)$};
\draw[->] (1cm,-2.25cm) -- (4cm,-1.5cm) node[midway,above,sloped]{$a, (1,0,r,r,1)$};

\draw[xshift= 5cm,yshift= -3cm](0.1,0) node [name=post]{$s', (3,0,5,0,9)$};
\draw[->] (1cm,-2.25cm) -- (4cm,-3cm) node[midway,below,sloped]{$a, (1,0,r,r,1)$};
\draw [<->,line width=1pt,xshift=2.35cm,yshift=-2.65cm,rotate=-15] (0,0) -- ++(0cm,.3cm) -- ++(+0.32cm,0cm) -- ++(0cm,-0.3cm);
\draw[xshift=4.05cm,yshift=-2.25cm](0,0) node{\small{swap the counter values}};



\foreach \y in {0cm,-0.4cm,...,-1.6cm} {
  \draw (-1.6cm,\y-1.4cm) +(-.3,-.2) rectangle ++(.3,.2);
  }

\draw (-1.6cm,-1.4cm) node{$2$};
\draw (-1.6cm,-1.8cm) node{$5$};
\draw (-1.6cm,-2.2cm) node{$\Null$};
\draw (-1.6cm,-2.6cm) node{$1$};
\draw (-1.6cm,-3.0cm) node{$\Null$};

\foreach \y in {0cm,-0.4cm,...,-1.6cm} {
  \draw (6.8cm,\y-0.6cm) +(-.3,-.2) rectangle ++(.3,.2);
  }

\draw (6.8cm,-0.6cm) node{$2$};
\draw (6.8cm,-1.0cm) node{$5$};
\draw (6.8cm,-1.4cm) node{$2$};
\draw (6.8cm,-1.8cm) node{$\Null$};
\draw (6.8cm,-2.2cm) node{$\Null$};

\foreach \y in {0cm,-0.4cm,...,-1.6cm} {
  \draw (7.8cm,\y-1.8cm) +(-.3,-.2) rectangle ++(.3,.2);
  }

\draw (7.8cm,-1.8cm) node{$3$};
\draw (7.8cm,-2.2cm) node{$3$};
\draw (7.8cm,-2.6cm) node{$5$};
\draw (7.8cm,-3.0cm) node{$\Null$};
\draw (7.8cm,-3.4cm) node{$\Null$};

\draw[->] (6.2cm,-1.5cm) -- (6.4cm,-1.4cm);
\draw[->] (6.2cm,-3cm) -- (7.4cm,-2.7cm);

\draw (.5cm,-1.0cm) node{{\bf choice of $2$:}};
\draw (.5cm,-3.5cm) node{{\bf choice of $3$:}};

\end{tikzpicture}
  \caption{An example of the construction of simulating transitions for a transition from $s$ to $s'$ labeled with an effect $(1,0,*2,r,1)$. In this example, the simulating transitions start from the state $s$ with the parent pointers set to the values $2,5,\Null,1,\Null$ (for the counters $1, 2, \dots, 5$) and the counter values set to $(2,5,9,1,8)$. The parent pointer values and the counter values only illustrate the parent pointer manipulations and the application of the effects, they might differ in actual runs.}
  \label{Fig:copy-encoding}
\end{figure}
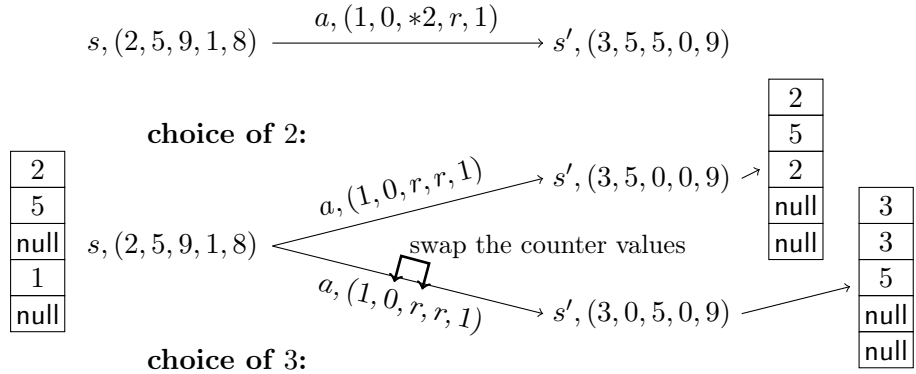

\paragraph{\bf Proof of Correctness.}
Intuitively, the choice of a counter in the copy instruction tells that the value in this counter will be destroyed by a reset or overwritten by a copy instruction later than in the counter which was not chosen. The structure of the copies is captured by the parent pointers in the following sense. If the counter $i$ points to the counter $j$ then $i$ contains an immediate copy of $j$ (but possibly modified by increments) and its value will be destroyed earlier than the value in $j$. The automaton ends in the error sink if it witnesses a violation of some of these implicit claims, i.e., the value in the counter $i$ is destroyed earlier than the value in the counter $j$.

First, we formalize the concept of the evolution of a value and define the corresponding runs. Then we show existence of corresponding accepting runs. Later on we use the fact that the parent pointers along the simulating traces have a special structure to show the correctness of the simulation.


\begin{defi}
\label{Def:value-trace}
  For a path $\sigma$ of length $|\sigma|$ in the extended R-automaton (considered as a graph) with $n$ counters and for two natural numbers $1\leq i < j \leq |\sigma|$, a total function $vt:\{i, i+1, \dots, j\} \arrow{} \{1, \dots, n\}$ is a \emph{value trace} if for all $k$ such that $i \leq k < j$, $t$ is the effect on the transition between the $k$-th and $k+1$-st state on $\sigma$, $vt(k) = a$, $vt(k+1) = b$, the following holds: if  $a\neq b$ then $\pi_b(t) = *a$ and if $a = b$ then $\pi_b(t) \in \{0, 1\}$.
\end{defi}

\forget{
\begin{defi}
  For a path $p$ of length $|p|$ and two natural numbers $1\leq i < j \leq |p|$, a total function $vt:\{i, i+1, \dots, j\} \arrow{} \{1, \dots, n\}$ is a \emph{value trace} if for all $k$ such that $i \leq k < j$, $t$ is the effect on the transition between the $k$-th and $k+1$-st state on $p$, $vt(k) = a$, $vt(k+1) = b$, the following holds:
  \begin{enumerate}[$\bullet$]
  \item if  $a\neq b$ then $\pi_b(t) = *a$ and
  \item if $a = b$ then $\pi_b(t) \in \{0, 1\}$.
  \end{enumerate}
\end{defi}
}

A value trace follows a value from some time point during its evolution (increments, copying) in an extended R-automaton. A value trace ends before the value is overwritten by a copy instruction or reset. We also talk about a value trace along a run. Then we mean a value trace along a path which has induced the run. We order value traces by the set inclusion on their domains (e.g., $vt: \{2,3\} \arrow{} \{1, \dots, n\}$ is smaller than $vt: \{2,3,4\} \arrow{} \{1, \dots, n\}$ regardless of the actual function values). We define the \emph{length} of a value trace as the size of its domain.

Now we define the correspondence between accepting runs in an extended R-automaton $R$ and in its corresponding R-automaton $\hat{R}$. We say that a run $\rho$ of $R$ over $w$ and a run $\rho'$ of $\hat{R}$ over $w$ are corresponding if for all $i$ the $i$-th transitions of $\rho, \rho'$ are obtained by executing the transitions $s \arrow{a,t} s'$ and $s \arrow{a,t',sp} s'$, where $s \arrow{a,t',sp} s'$ is a simulating transition of $s \arrow{a,t} s'$. We show that for each accepting run of one automaton there is an accepting corresponding run of the other automaton. It follows immediately from the definitions that for each accepting run of $\hat{R}$ there is exactly one accepting corresponding run of $R$.

\forget{
\begin{lemma}
  For each accepting run of $\hat{R}$ there is exactly one accepting corresponding run of $R$.
\end{lemma}

\begin{proof}
  Immediately from the definitions.
\end{proof}
}

The other direction is more complicated, because we have to show that $\hat{R}$ can choose correct values for non-deterministic choices in the copy instruction so that it does not end up in the error sink. For each accepting run $\rho$ of $R$, we construct an accepting run $\rho'$ of $\hat{R}$ as follows. We label each counter $j$ in the $k$-th state of $\rho$ (for all $k \leq |\rho|$) by the length of a maximal value trace $vt$ with domain being a subset of $\{k, k+1, \dots, |\rho|\}$ and $vt(k) = j$ (this label is called \emph{expectancy}). $\hat{R}$ takes the simulating transition for each transition of $\rho$ (according to the rules above) and when it has to choose between $i$ and $j$ ($e_i = *j$) along a transition ending in the $k$-th state, then it chooses $i$ if and only if the expectancy of $i$ in $k$ is greater than the expectancy of $j$ in $k$ \emph{(expectancy rule)}. We show that this is a valid definition, i.e., the corresponding run of $\hat{R}$ does not end up in the error sink. The main step in the proof is to show that the parent pointers always point to the counters with expectancy which is greater than or equal to the expectancy of the counter which owns the parent pointer.

\begin{lemma}
\label{Lemma:existence-of-corr-runs}
  For each accepting run $\rho$ of $R$ there is an accepting corresponding run $\rho'$ of $\hat{R}$.
\end{lemma}

\proof{
  We prove by induction that for each prefix of $\rho$ there is a simulating run which does not contain the error state such that for any state along $\rho'$ and any two counters $i,j$ in this state, if the parent pointer of $i$ points to $j$ then the expectancy of $j$ is not smaller than that of $i$. Such a simulating run for  $|\rho|$ will also be accepting.

  The basic step (i.e., the prefix length is $0$) is trivial.
  For the induction step, let us assume that there is a simulation of the prefix of length $k$ satisfying IH. To simulate the $k+1$-st transition, we follow the expectancy rule.

  Because of the induction hypothesis and the definition of expectancy, there are always simulating transitions (and not a transition leading to the error sink). If there is a copy instruction $e_i = *j$ in the transition, the non-deterministic choice is performed according to the $vt$ function, so the result again satisfies the induction hypothesis. The transfer of the parent pointers does not violate it either, because expectancy of $j$ in $k$ is equal to $1$ plus the maximum of the expectancies of $i$ and $j$ in $k+1$.
  The resets do not establish any new parent pointers, so the result again satisfies the induction hypothesis.
  The other instructions result in decrementing the expectancy, which preserves the induction hypothesis for all the pointers inherited from the previous state as well as for the pointers changed by the copy instruction.
\forget{
  For the induction step, let us assume that there is a simulation of the prefix of length $k$ satisfying the induction hypothesis. To simulate the $k+1$-st transition, we follow the expectancy rule.

  Copy: we do not go to the error sink, because of IH. Choice according to the $vt$ function, so satisfying IH. The transfer of the parent pointers is OK (at this point, we compare with the expectations in the $k$-th state and later show that it was OK), because expectancy of $j$ in $k$ is equal to $1$ plus the maximum of the expectancies of $i$ and $j$ in $k+1$.

  Resets: we do not go to the error sink, because of IH. No contribution to the parent pointer structure.

  Other: decrement of the expectancy, which makes it OK for all the pointers from before and also for the pointers changed by the copy.
}
\qed
}

Let us introduce the parent pointer relation $\ppr$ for a state of $\hat{R}$ as a relation on counters where $i \ppr j$ if and only if the parent pointer of $i$ is set to $j$.

\begin{lemma}
\label{Lemma:antireflexivity}
  Let $\rho$ be a run of $\hat{R}$. The transitive closure of $\ppr$ is antireflexive in all states of $\rho$.
\end{lemma}

\proof{
  We prove by induction that for each prefix of $\rho$, the transitive closure of $\ppr$ is antireflexive in all states of the prefix.

  The basic step is trivial, $\ppr$ is empty in $s_0$.
  For the induction step, we need to check that a single transition does not violate the antireflexivity. If the transition leads to the error sink then $\ppr$ is not changed. Otherwise, it is a simulating transition defined by the rules above. The resets make $\ppr$ smaller and $0, 1$ do not change it. In the copy instruction $e_i = *j$, we introduce one new pointer, but we know that nothing points to $i$, because of the condition on creating the simulating transitions and the fact that the parent pointers of all reset counters are set to $\Null$. In the first case ($j$ has been chosen), we set $i$'s parent pointer to $j$, which cannot introduce a loop, since nothing points to $i$. In the second case ($i$ has been chosen), since we have redirected all the pointers pointing to $j$ to $i$, there is nothing pointing to $j$ and newly introduced $j\ppr i$ cannot create a loop. Also, since there was nothing pointing to $i$ previously, the only pointers pointing to $i$ now are those that previously pointed to $j$. 
  \qed

}

This leads to the following definition of ranks.
  For a counter $i$ in a state $s$ of $\hat{R}$ we define $\rank(s,i)$ inductively by $\rank(s,i) = 0$ if the parent pointer of $i$ in $s$ is $\Null$ and $\rank(s,i) = \rank(s,j)+1$ if $i\ppr j$ in $s$.
\forget{
\begin{defi}
  For a counter $i$ in a state of $\hat{R}$ we define $\rank(i)$ inductively by $\rank(i) = 0$ if the parent pointer of $i$ is $\Null$ and $\rank(i) = \rank(j)+1$ if $i\ppr j$.
\end{defi}
}
From Lemma~\ref{Lemma:antireflexivity}, we have that the ranks are well-defined and it follows directly from the definition that the rank of a counter is always bounded by the number of the counters. Now we formulate a lemma saying that the ranks never decrease along a value trace.

\begin{lemma}
\label{Lemma:ranks}
  Let $\rho$ be a run of $\hat{R}$ and $vt$ be a value trace. Then for $k\leq l$ such that $vt(k), vt(l)$ are defined, $\rank(s_k, vt(k)) \leq \rank(s_l, vt(l))$.
\end{lemma}

\proof{
  We show this claim by induction on $l-k$.
  The basic step is that $l = k$ and then $\rank(s_k, vt(k)) = \rank(s_l, vt(l))$.
  For the induction step we have two cases. If the transition leads to the error sink then $\ppr$ is not changed and therefore the ranks do not decrease. Otherwise, it is a simulating transition defined by the rules above. Because of the condition on creating the simulating transitions, we never decrease any rank by a reset. The instructions $0, 1$ also do not decrease any rank. Copy increases the rank of the branch with smaller expectancy (and the counter is reset) and keeps the rank for the branch with bigger expectancy (the one which keeps the value) unchanged. Because of the careful manipulation with the pointers, no ranks which depend on the rank of the longer branch change either.\qed
}

\forget{
The correctness of the reduction is stated in the following lemma.

\begin{lemma}
  Let $R$ be an extended R-automaton with at most one copy instruction in each effect and $\hat{R}$ be the simulating R-automaton constructed as above. There is a $B$ such that $L_B(R) = \Sigma^*$ if and only if there is a $B'$ such that $L_B'(\hat{R}) = \Sigma^*$.
\end{lemma}

\proof{
  The "if" direction: from Lemma~\ref{Lemma:existence-of-corr-runs} we have that for each accepting run $\rho$ of $R$ there is a corresponding accepting run $\rho'$ of $\hat{R}$. It is easy to see from the construction that for all $k \leq |\rho|$, the counter values in the $k$-th state of $\rho'$ are not bigger than the counter values in the $k$-th state of $\rho$. The counter instructions are simulated faithfully, except for possibly more resets (in the copy instructions) along $\rho'$.

  The "only if" direction: for each $B$, let us take a counterexample $w$ to the universality of $R$ in the $n\cdot B$ semantics ($n$ is the number of the counters). Any accepting run $\rho'$ of $\hat{R}$ over $w$ must satisfy Lemma~\ref{Lemma:ranks} (because it corresponds to an accepting run $\rho$ of $R$ over $w$). Let $vt$ be a maximal value trace for a value which exceeds $n\cdot B$ in $\rho$. We study the evolution of this value in $\rho'$. It is simulated faithfully except for some possible resets in the copy instructions. But for each such reset, the rank of the counter strictly decreases. Therefore, there can be at most $n-1$ such resets and there must be a state in which this value exceeds $B$.
}
}

The main property of the reduction is stated in the following lemma. The correctness of Lemma~\ref{Lemma:RC-automata} is then a direct corollary of this lemma.

\begin{lemma}
\label{Lemma:decidability-of-copy}
  Let $R$ be an extended R-automaton with $n$ counters and with at most one copy instruction in each effect and $\hat{R}$ be the simulating R-automaton constructed as above. For each $B$ and for each word $w$, $w\in L_B(R) \Rightarrow w\in L_B(\hat{R})$ and $w\in L_B(\hat{R}) \Rightarrow w\in L_{n\cdot B}(R)$.
\end{lemma}

\proof{
  The first implication: we know from Lemma~\ref{Lemma:existence-of-corr-runs} that for each accepting run $\rho$ of $R$ over $w$ there is a corresponding accepting run $\rho'$ of $\hat{R}$ over $w$. It follows directly from the construction that for all $k \leq |\rho|$, the counter values in the $k$-th state of $\rho'$ are bounded by the counter values in the $k$-th state of $\rho$. All instructions are simulated faithfully except for replacing copy instructions by resets along $\rho'$.

  The second implication: by contraposition, let us for each $B$ consider a word $w$ such that $w\notin L_{n\cdot B}(R)$. Any accepting run $\rho'$ of $\hat{R}$ over $w$ must satisfy Lemma~\ref{Lemma:ranks}. Let $vt$ be a maximal value trace for a value which exceeds $n\cdot B$ in $\rho$. We study the evolution of this value in $\rho'$. It is simulated faithfully except for some possible resets in the copy instructions. But for each such reset, the rank of the counter strictly increases. Therefore, there can be at most $n-1$ such resets and there must be a state in which this value exceeds $B$.\qed
}


Now we show that the result holds also for extended R-automata with any number of copying in each step. Let us view the relation "$i$ is copied to $j$" induced by an effect $t$ as a directed graph (counters are nodes, there is an edge from $i$ to $j$ if $\pi_j(t) = *i$). Because each node can have at most one incoming edge, such a graph is a collection of simple loops with isolated paths outgoing from them (nodes with no incoming edge are considered as degenerated loops). We can split application of such an effect $t$ into an equivalent sequence of effects with at most one copy instruction and some swapping of the values and the parent pointers as follows. First, we perform $\hat{t}$ (all increments and resets). Then we pick one of the counters $j$ such that $j$ has no outgoing edge and it has an (exactly one) incoming edge from $i$. We copy the value of $i$ to $j$ and leave all other counters unchanged, which can be described by the effect $(0, \dots, *i, \dots, 0)$, where $*i$ is on the $j$-th position. Then we remove the edge connecting $i$ and $j$ and continue to pick another such counter. When there is no node $j$ with no outgoing edge and with an incoming edge, there still might be loops in the copying graph. We simply swap the counter values and the parent pointers in the loops. Because of the order in which we have copied the counters, the effect of this sequence of transitions with at most one copy instruction and swaps is the same as that of the original transition. Also, the correctness does not depend on the order in which we choose the edges. A careful analysis shows that this sequence of transitions can be encoded into one simulating transition in R-automata with value swapping and parent pointers.

\subsection{Limiting Maxima in Extended R-automata}
\label{Sec:Reduction-maxplus-R}


\forget{
\begin{lemma}[\cite{aky08rc-automata}]
\label{Lemma:RC-automata}
  Limitedness and universality are decidable for extended R-automata.
\end{lemma}
}

Let for a state $\langle s, (\bar{P}, \bar{M}, \bar{N}), \lesssim \rangle$ in a run of an extended R-automaton with $n$ counters, the $N$-value ($M$-value, $P$-value) of this state be $\max\{N_i | 1\leq i \leq n\}$ ($\max\{M_i | 1\leq i \leq n\}$, $\max\{P_i | 1\leq i \leq n\}$, respectively). Let for a run $\rho$ of this automaton, the $N$-value ($M$-value, $P$-value) of the run be the maximum state $N$-value ($M$-value, $P$-value) over all states along the run. We denote this value by $N(\rho)$ ($M(\rho)$, $P(\rho)$).

\begin{lemma}
\label{Lemma:counter-maxs}
  Let $R$ be an extended R-automaton with $n$ counters and let $B\in\Nat$. For all runs $\rho$ of $R$, if $N(\rho) \leq B$ then $M(\rho) \leq B^n$.
\end{lemma}

\proof{
  We show a stronger claim, namely that if a run $\rho$ starts in a state with the $M$-value equal to $b$ and $N(\rho) \leq B$ then $M(\rho) \leq b + B^n$, by induction on the number of counters $n$. The basic step ($n=1$) is trivial, because Step~\ref{Sem:max} will never change the counter value and thus $M(\rho) = N(\rho) \leq B$.

  Let us assume that the claim holds for automata with $n$ counters. We show that it holds for automata with $n+1$ counters. Let us fix a run $\rho$ and a $B \in \Nat$. Let us without loss of generality assume that the counter which reaches the greatest $M$ value is the counter $n+1$. First, we argue that there is an extended R-automaton and a run of this automaton starting with the same counter values as $\rho$ which has the same $M$-value as $\rho$, along which the counter $n+1$ is never updated by a copy instruction and never reset.

  The argument for the copy instructions is straightforward, each copy instruction makes the source and the target counter equivalent both in the values which it contains (Step~\ref{Sem:effects}) and in the preorder $\lesssim$ (Step~\ref{Pre:copy}). Therefore, we can permute the instructions in the effects (intuitively, rename the counters) in the prefix of the run leading to the copy instruction so that the value is accumulated in the counter $n+1$ and then copied to the other counter.

  If the counter $n+1$ is reset then its values can be incremented only by $1$ and via the $\max$ operation with other counters which are reset later. This follows from the fact that $n+1$ is a minimal element of $\lesssim$ after it is reset. This is the same situation as if the run started with all counter values equal to zero ($\bar{C}_0$) and $\lesssim$ empty.

  Therefore, the counter $n+1$ can be updated only by $1$ and $0$ (where $0$ does not increase $M_{n+1}$ and there can be at most $B$ increases by $1$) and $M_{n+1}$ can be increased by the $\max$ operation. We show that $M_{n+1}$ can grow by at most $B^n$ between any two increments by $1$.

  Between any two increments by $1$, the value $M_{n+1}$ can grow only by application of the $\max$ operation with the counters $i$ such that $i \lnsim n+1$ (Step~\ref{Sem:max}). These counters cannot make use of the counter $n+1$ (cannot increase their $M$-values more than if there was no counter $n+1$). The only way for a counter $i$ to use the counter $n+1$ is to apply $i = *(n+1)$, but this would set $i\simeq n+1$ (Step~\ref{Pre:copy}). To set $i \lnsim n+1$ back again, we would have to reset $i$ (instruction $r(\{n+1\} \cup A)$) or copy some other counter $j$ such that $j \lnsim n+1$ into $i$ (instruction $i = *j$) (follows from Step~\ref{Sem:preorder}). But this would have the same effect as if $i$ was updated by $0$ until this state and then reset or copied. Hence, the claim that $M_{n+1}$ can grow by at most $B^n$ between any two increments by $1$ follows from IH.\qed

}

Now we show that the $P$-values are bounded by an exponent of the $M$-values.

\begin{lemma}
\label{Lemma:counter-sums}
  Let $R$ be an extended R-automaton with $n$ counters and let $B\in\Nat$. For all runs $\rho$ of $R$, if $M(\rho) \leq B$ then $P(\rho) \leq 2^B$.
\end{lemma}

\proof{
  We show by induction on the length of the run that for all states $\langle s, (P,$ $M$, $N), \lesssim \rangle$ along the run and for all $1\leq i \leq n$, $P_i \leq 2^{M_i}$. The basic step is trivial. We check that the claim is preserved by every update of the counters. Let us denote the values before the transition by unprimed letters $P, M$ and after the effect takes place with primed letters $P', M'$. Let the instruction (update) applied to the counter $i$ be:
  \begin{enumerate}[$\bullet$]
    \item 0 : The values of the counters do not change, the claim holds from IH.
    \item 1 : We have that $P'_i = P_i+1$, $M_i' = M_i+1$. From IH, we know that $P_i\leq 2^{M_i}$. From this we have that $P_i' = P_i+1 \leq 2^{M_i} + 1$. Because $2^{M_i} \geq 1$ for all $M_i\geq 0$, we have that $2^{M_i} + 1 \leq 2 \cdot 2^{M_i} = 2^{M_i+1} = 2^{M_i'}$.
    \item r : This case is clear, $P_i' = M_i' = 0$.
    \item *j : The claim follows from IH.
    \item $\max$ : Let us discuss one application of the $\max$ operation (Step~\ref{Sem:max}) where the value of $P_i'$ is increased (if it is not the case then the claim holds from IH). If $k,l \lnsim i$ then $P'_i = \max\{P_i, P_k + P_l\} = P_k + P_l$ and $M'_i = \max\{ M_i, M_k+1, M_l+1\}$. Without loss of generality, let us assume that $P_k \geq P_l$. Thus, $P'_i \leq 2\cdot P_k \leq 2\cdot 2^{M_k} = 2^{M_k+1}$. Since $M'_i \geq M_k + 1$, we have that $2^{M_k + 1} \leq 2^{M'_i}$.

  \end{enumerate}

  \noindent Update of all counters along each transition consists only of these updates.\qed
}

\section{Encoding of Timed Automata to Extended R-automata}
\label{Sec:Encoding-TA-RC}

\revision{MOTIVATION FOR COUNTERS AND PREORDER}

Now we are ready to show the translation of timed automata into extended R-automata. Intuitively, we equip the region graph induced by a given timed automaton with counters whose values are updated as we move along a path in the region graph. 
The constructed extended R-automaton is equipped with two counters, $C_{xy}$ and $C_{yx}$, for each pair of clocks $x, y$. These counters keep the information about the minimal distances between the fractional parts of the clocks. The distance is not characterized in an absolute manner, but relatively to a sampling unit $\epsilon$. Let the counter values be obtained after following a path in the region graph. The counters say how many $\epsilon$'s at least have to be there between the fractional parts of two clocks in any state reachable by a concrete run of the timed automaton along this path in the region graph. 

If the fractional parts coincide then both $C_{xy}$ and $C_{yx}$ are equal to $0$. If the fractional parts are not equal then $C_{xy}$ contains a lower bound (as a number of $\epsilon$ steps) on the distance from $x$ to $y$ ($\numodels{xy}$) and $C_{yx}$ a lower bound on the distance from $y$ to $x$ ($\numodels{yx}$). This lower bound is also tight -- up to factor $2$. 
If the extended $R$-automaton reaches a state where $C_{xy}$ contains $12$ along some path then each run in an $\epsilon$-sampled semantics along the corresponding path in the region graph will end up in a state where $\numodels{xy} \geq 12\cdot \epsilon$. If $\epsilon = 0.01$ then the distance between the fractional parts has to be at least $0.12$. If $\epsilon = 0.1$ then this state is unreachable along this region graph path, because the difference
between the fractional parts has to be always smaller than $1$. Also, states where $\numodels{xy} \geq 2\cdot C_{xy} \cdot \epsilon$ holds for all clocks $x,y$ can be reached along the corresponding path in $\epsilon$-sampled semantics.

If the extended $R$-automaton is limited then we can choose a sufficiently small $\epsilon$ such that for each untimed word there is an accepting state which can be reached in $\epsilon$-sampled semantics (while reading this word). On the other hand, if the extended $R$-automaton is unlimited then for each $\epsilon$ we can pick a word which is accepted only with some counter exceeding $1/\epsilon$. This means that there will not be any runs in $\epsilon$-sampled semantics accepting this word.

The following examples illustrate how do we update the counters. First, we look at counter incrementing. If we start in the region where $0 < x = y < 1$ then 
the counters $C_{xy}$ and $C_{yx}$ are equal to $0$. Assume that the automaton resets the clock $x$. Then the distance between $x$ and $y$ has to be at least $\epsilon$ (and the same holds for the distance from $y$ to $x$). Hence, we increment both counters. After a (symbolic) time pass transition in the region graph, we come to the region where $0 < x < y < 1$. If the automaton now resets $x$ again then the distance from $x$ to $y$ has to be at least $2\cdot\epsilon$, while the lower bound on the distance from $y$ to $x$ can be arbitrarily small (but at least $\epsilon$). Therefore, we increment the counter $C_{xy}$, reset the counter $C_{yx}$ and immediately increment it.

Secondly, we describe a scenario where we need counter copying. Assume that the clocks $x$ and $y$ have different fractional parts and there is some value in the counter $C_{xy}$, say $54$. We can reset clocks $u$ and $v$ within one time unit so that the fractional parts of $u$ and $x$ are the same and that the fractional parts of $v$ and $y$ are the same. Then we know that the difference between the fractional parts of $u$ and $v$ has to be at least $54 \cdot \epsilon$. To remember this fact, we copy the information from the counter $C_{xy}$ to the counter $C_{uv}$. Note that it is not enough to track both distances with one counter, because these distances can from now on develop independently.

Finally, the last example motivates the maximum operation. Assume that we have three clocks $x,y$ and $z$ in a region where $0 < x < y < z < 1$ and $C_{xy} = 10, C_{yz} = 13$. This means that the difference between the fractional parts of $x$ and $y$ has to be at least $10 \cdot \epsilon$ and the difference between the fractional parts of $y$ and $z$ has to be at least $13 \cdot \epsilon$. It follows that the difference between the fractional parts of $x$ and $z$ has to be at least $23 \cdot \epsilon$. This fact has to be reflected in the value of $C_{xz}$, which has to be at least $C_{xy} + C_{yz}$. If the distance between the clocks $y$ and $z$ increases and the counter $C_{yz}$ is incremented (as described above) then we have to update the counter $C_{xz}$ so that it contains the value $\max\{ C_{xz}, C_{zy} + C_{yz} \}$. Symmetrically, the same holds for increments of the counter $C_{xy}$. In our model, this is ensured by maintaining the pre-order $\lesssim$ in such a way that $C_{xy} \lnsim C_{xz}$ and $C_{yz} \lnsim C_{xz}$ hold if and only if $xy$ and $yz$ are subintervals of $xz$ (formally, $D \models \overline{xyz}$). Then the automaton automatically sets $C_{xz}$ to $C_{xy} + C_{yz}$ if this value becomes greater than $C_{xz}$. We update the pre-order $\lesssim$ along the transitions by supplying additional information to reset operations. We track the "is subinterval of" relation for intervals between pairs of clocks, which is available directly from the region.


\revision{END OF MOTIVATION FOR COUNTERS AND PREORDER}

\revision{Restricting max operation}

In order to avoid summing up overlapping intervals 
we restrict the max operation of extended R-automata in the following way. The $P$ counters are updated by a sum of two other $P$ counters only if the other two counters do not have a lower bound in the $\lesssim$ ordering. Formally, Step~\ref{Sem:max} in the unparameterized semantics definition now reads:

\begin{enumerate}[$\bullet$]
\item[\ref{Sem:max}.] Repeat the following until a fixed point is reached:
    \begin{enumerate}[$-$]
        \item if $i \lnsim j$ then set $M_j' = \max\{M_j', M_i' + 1\}$ and
        \item if $k,l \lnsim j$ and $\nexists m. m \lesssim k \wedge m \lesssim l$ then set $P_j' = \max\{P_j', P_k' + P_l'\}$
    \end{enumerate}
\end{enumerate}

We need this restriction in order to count each subinterval only once. As an example, consider a region where $0 < a < b < c < d < 1$. Then, $C_{ac} \lnsim C_{ad}$ and $C_{bd} \lnsim C_{ad}$, and therefore according to the original definition, $C_{ad}$ has to be at least as big as $C_{ac} + C_{bd}$. This includes the difference between fractional parts of clocks $b$ and $c$ (represented by the counter $C_{bc}$) twice. This is not possible with the new definition, because $C_{bc} \lesssim C_{ac}$ and $C_{bc} \lesssim C_{bd}$.

Clearly, Lemma~\ref{Lemma:counter-sums} holds also for this restriction, because the $P$-values will be always smaller than or equal to the $P$-values calculated according to the original definition.

\revision{End of: Restricting max operation}

The rest of this section is organized as follows. First, we describe how to translate a timed automaton with at most one clock reset along each transition into an extended R-automaton. Then we show three technical properties of the constructed extended R-automaton (Lemma~\ref{Lemma:constr-preorder}, Lemma~\ref{Lemma:counter-transitivity}, and Lemma~\ref{Lemma:counter-max-one-increase}). In the rest of this section we prove the correspondence between the counter values along runs of the extended R-automaton and the minimal distances between the fractional parts of the clock values in the timed automaton (Lemma~\ref{Lemma:TA-RC-correspondence} and Lemma~\ref{Lemma:cost-is-real}).

\paragraph{\bf Construction.} Let $G$ be the region graph induced by a given timed automaton $A'$ with at most one reset in each transition. We build an extended R-automaton $R$ from this region graph $G$. The extended R-automaton $R$ has a state corresponding to each node in the region graph $G$ and two auxiliary states for each edge in the region graph $G$ corresponding to a discrete transition (an edge labeled by $\alpha\in\Sigma$). The initial state is the state corresponding to the node $\langle q_0, \{\nu_0\}\rangle$. Accepting states are the states corresponding to the nodes $\langle q, D \rangle$, where $q\in F$. We introduce two counters $C_{xy}, C_{yx}$ for each pair of clocks $x,y\in\clocks$ where $x$ is different from $y$. 
We use only the $P$ values from the extended R-automaton and in the following we will refer to them simply by $C_{xy}, C_{yx}$.

Since encoding of a single edge might need to perform multiple counter updates, we introduce a sequence of three transitions and two auxiliary states between them for each edge in $G$ corresponding to a discrete transition of the timed automaton $A'$. These transitions are labeled by the same letter as the original edge. More precisely, let us have an edge in $G$ from $\langle q', D' \rangle$ to $\langle q, D \rangle$ labeled by $\alpha$, where $\alpha\in \Sigma$. Then we create two auxiliary states $q_1,q_2$ (these states are unique for this transition, formally we should write $q_{1}^{(q,D,\alpha,q',D')}, q_{2}^{(q,D,\alpha,q',D')}$, but without confusion, we skip the superscript) and three transitions from $\langle q', D' \rangle$ to $q_1$, from $q_1$ to $q_2$, and from $q_2$ to $\langle q, D \rangle$, all of them labeled by $\alpha$.

For edges corresponding to a time pass transition in $A$ (edges labeled by $\delta$), we introduce only one transition labeled by $\delta$ directly leading to the state corresponding to the target node. More precisely, let us have an edge from $\langle q', D' \rangle$ to $\langle q, D \rangle$ labeled by $\delta$ in $G$. We create a transition in $R$ from $\langle q', D' \rangle$ to $\langle q, D \rangle$ labeled by $\delta$. Later on, we show how to get rid of these transitions (and of the letter $\delta$) while preserving the counter bounds. In fact, the standard construction for showing that regular languages are closed under projection works, because transitions labeled by $\delta$ do not affect the counter values.

Now we show how to label the transitions by effects. The transitions labeled by $\delta$ and the transitions corresponding to an edge in $G$ from $\langle q', D' \rangle$ to $\langle q, D \rangle$ labeled by $\alpha$, $\alpha\in \Sigma$, where either $D = D'$ (no clock is reset) or a clock $x$ is reset such that $\Dpmodels{x}{=}{0}$ (the clock had zero fractional part before reset) are labeled by the effect $(0, \dots, 0)$ (all counters are left unchanged).


In other cases, we have transitions corresponding to an edge in $G$ from $\langle q', D' \rangle$ to $\langle q, D \rangle$ labeled by $\alpha$ where a clock with non-zero fractional part is reset. Let us denote this clock by $x$. These transitions are labeled by effects created according to the following four cases. Counters which are not mentioned are left unchanged (the instruction is $0$ on all three transitions). The instructions are denoted by pairs $C : e_1, e_2, e_3$, where $C$ is the counter, to which the instructions are applied and $e_1, e_2, e_3$ are the instructions ($e_i$ is a part of the effect on the $i$-th transition).


\begin{enumerate}[(1)]
  \item \label{Constr:regA} The region $D'$ has a clock $a$ with zero fractional part (depicted in Figure~\ref{Fig:regA}).
    \begin{enumerate}[(a)]
        \item $C_{ax}, C_{xa} : r(\emptyset), 0, 0$
        \item $u \neq a$. $C_{ux} : *C_{ua}, 0, 0$ and $C_{xu} :             *C_{au}, 0, 0$.
    \end{enumerate}

  \item \label{Constr:regB} The region $D'$ has clocks $a, d$ such that the fractional part of $a$ is smaller than or equal to the fractional part of $x$ and  the fractional part of $d$ is greater than or equal to the fractional part of $x$ (depicted in Figure~\ref{Fig:regB}).
    \begin{enumerate}[(1)]
        \item $u \neq a$. $C_{xu} : *C_{au}, 1, 0$
        \item $u \neq d$. $C_{ux} : *C_{ud}, 1, 0$
        \item $C_{xa} : 0, r(\{C_{da}, C_{xu} | \forall u \neq a,x \}), 1$
        \item $C_{dx} : 0, r(\{C_{da}, C_{ux} | \forall u \neq d,x \}), 1$
    \end{enumerate}

  \item \label{Constr:regC} The clock $x$ has strictly smaller fractional part than other clocks in $D'$ (depicted in Figure~\ref{Fig:regC}). We denote a clock with the smallest fractional part greater than the fractional part of $x$ by $a$ and a clock with the greatest fractional part by $d$.
    \begin{enumerate}[(1)]
        \item $u \neq x$. $C_{xu} : 1, 0, 0$
        \item $u \neq d$. $C_{ux} : *C_{ud}, 1, 0$
        \item $C_{dx} : 0, r(\{C_{da}, C_{ux} | \forall u \neq d,x \}), 1$
    \end{enumerate}

  \item \label{Constr:regD} The clock $x$ has strictly greater fractional part than other clocks in $D'$ (depicted in Figure~\ref{Fig:regD}). We denote a clock with the greatest fractional part smaller than the fractional part of $x$ in $D'$ by $d$ and a clock with the smallest fractional part in $D'$ by $a$.
    \begin{enumerate}[(1)]
        \item $u \neq x$. $C_{ux} : 1, 0, 0$
        \item $u \neq a$. $C_{xu} : *C_{au}, 1, 0$
        \item $C_{xa} : 0, r(\{C_{da}, C_{xu} | \forall u \neq a,x \}), 1$

    \end{enumerate}
\end{enumerate}

  \begin{figure}[htbp]
    \centering
    \begin{tikzpicture}[>=stealth, scale=1.3]

      \path(-2,0) coordinate (0);
      \path(2,0) coordinate (1);
      \path(-2,.1) coordinate (0u);
      \path(-2,-.1) coordinate (0l);
      \path(2,.1) coordinate (1u);
      \path(2,-.1) coordinate (1l);
      \path(-1,.1) coordinate (bu);
      \path(-1,-.1) coordinate (bl);
      \path(-.2,.3) coordinate (xu);
      \path(-.4,0) coordinate (xl);
      \path(.3,.1) coordinate (cu);
      \path(.3,-.1) coordinate (cl);
      \path(1,.1) coordinate (du);
      \path(1,-.1) coordinate (dl);
      \draw(-2,-.3) node [name=zero]{$0$};
      \draw(2,-.3) node [name=one]{$1$};
      \draw(-2,.3) node [name=x]{$x,a$};
      \draw(-1,.3) node [name=b]{$b$};
      \draw(-.14,.5) node [name=b]{$x'$};
      \draw(.3,.3) node [name=c]{$c$};
      \draw(1,.3) node [name=d]{$d$};

      \draw (0)--(1);
      \draw (0u)--(0l);
      \draw (1u)--(1l);
      \draw (bu)--(bl);
      \draw[->,dashed] (xu)--(xl);
      \draw (cu)--(cl);
      \draw (du)--(dl);

    \end{tikzpicture}
    \caption{The region $D'$ has a clock $a$ with zero fractional part. The letter $x'$ denotes the position of the clock $x$ in $D'$ (before it was reset).}
    \label{Fig:regA}
  \end{figure}
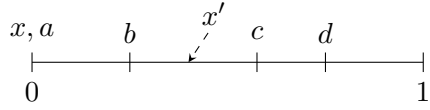

  \begin{figure}[htbp]
    \centering
    \begin{tikzpicture}[>=stealth, scale=1.3]

      \path(-2,0) coordinate (0);
      \path(2,0) coordinate (1);
      \path(-2,.1) coordinate (0u);
      \path(-2,-.1) coordinate (0l);
      \path(2,.1) coordinate (1u);
      \path(2,-.1) coordinate (1l);
      \path(-1.4,.1) coordinate (au);
      \path(-1.4,-.1) coordinate (al);
      \path(-1,.1) coordinate (bu);
      \path(-1,-.1) coordinate (bl);
      \path(-.2,.3) coordinate (xu);
      \path(-.4,0) coordinate (xl);
      \path(.3,.1) coordinate (cu);
      \path(.3,-.1) coordinate (cl);
      \path(1.2,.1) coordinate (du);
      \path(1.2,-.1) coordinate (dl);
      \draw(-2,-.3) node [name=zero]{$0$};
      \draw(2,-.3) node [name=one]{$1$};
      \draw(-2,.3) node [name=x]{$x$};
      \draw(-1,.3) node [name=b]{$b$};
      \draw(-1.4,.3) node [name=a]{$a$};
      \draw(-.14,.5) node [name=b]{$x'$};
      \draw(.3,.3) node [name=c]{$c$};
      \draw(1.2,.3) node [name=d]{$d$};

      \draw (0)--(1);
      \draw (0u)--(0l);
      \draw (1u)--(1l);
      \draw (au)--(al);
      \draw (bu)--(bl);
      \draw[->,dashed] (xu)--(xl);
      \draw (cu)--(cl);
      \draw (du)--(dl);

    \end{tikzpicture}
    \caption{The region $D'$ has no clock with zero fractional part and the fractional part of $x$ is neither strictly smaller nor strictly greater than all other clocks. The letter $x'$ denotes the position of the clock $x$ in $D'$ (before it was reset).}
    \label{Fig:regB}
  \end{figure}
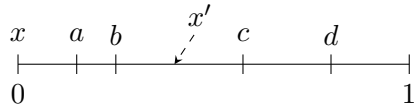

  \begin{figure}[htbp]
    \centering
    \begin{tikzpicture}[>=stealth, scale=1.3]

      \path(-2,0) coordinate (0);
      \path(2,0) coordinate (1);
      \path(-2,.1) coordinate (0u);
      \path(-2,-.1) coordinate (0l);
      \path(2,.1) coordinate (1u);
      \path(2,-.1) coordinate (1l);
      \path(-.8,.1) coordinate (bu);
      \path(-.8,-.1) coordinate (bl);
      \path(-1.5,.3) coordinate (xu);
      \path(-1.6,0) coordinate (xl);
      \path(.3,.1) coordinate (cu);
      \path(.3,-.1) coordinate (cl);
      \path(1.2,.1) coordinate (du);
      \path(1.2,-.1) coordinate (dl);
      \draw(-2,-.3) node [name=zero]{$0$};
      \draw(2,-.3) node [name=one]{$1$};
      \draw(-2,.3) node [name=x]{$x$};
      \draw(-.8,.3) node [name=b]{$a$};
      \draw(-1.44,.5) node [name=b]{$x'$};
      \draw(.3,.3) node [name=c]{$b$};
      \draw(1.2,.3) node [name=d]{$d$};

      \draw (0)--(1);
      \draw (0u)--(0l);
      \draw (1u)--(1l);
      \draw (bu)--(bl);
      \draw[->,dashed] (xu)--(xl);
      \draw (cu)--(cl);
      \draw (du)--(dl);

    \end{tikzpicture}
    \caption{The clock $x$ has the smallest (strictly) fractional part in the region $D'$. The letter $x'$ denotes the position of the clock $x$ in $D'$ (before it was reset).}
    \label{Fig:regC}
  \end{figure}

  \begin{figure}[htbp]
    \centering
    \begin{tikzpicture}[>=stealth, scale=1.3]

      \path(-2,0) coordinate (0);
      \path(2,0) coordinate (1);
      \path(-2,.1) coordinate (0u);
      \path(-2,-.1) coordinate (0l);
      \path(2,.1) coordinate (1u);
      \path(2,-.1) coordinate (1l);
      \path(-1,.1) coordinate (bu);
      \path(-1,-.1) coordinate (bl);
      \path(1.5,.3) coordinate (xu);
      \path(1.6,0) coordinate (xl);
      \path(.3,.1) coordinate (cu);
      \path(.3,-.1) coordinate (cl);
      \path(1,.1) coordinate (du);
      \path(1,-.1) coordinate (dl);
      \draw(-2,-.3) node [name=zero]{$0$};
      \draw(2,-.3) node [name=one]{$1$};
      \draw(-2,.3) node [name=x]{$x$};
      \draw(-1,.3) node [name=b]{$a$};
      \draw(1.54,.5) node [name=b]{$x'$};
      \draw(.3,.3) node [name=c]{$b$};
      \draw(1,.3) node [name=d]{$d$};

      \draw (0)--(1);
      \draw (0u)--(0l);
      \draw (1u)--(1l);
      \draw (bu)--(bl);
      \draw[->,dashed] (xu)--(xl);
      \draw (cu)--(cl);
      \draw (du)--(dl);

    \end{tikzpicture}
    \caption{The clock $x$ has the greatest (strictly) fractional part in the region $D'$. The letter $x'$ denotes the position of the clock $x$ in $D'$ (before it was reset).}
    \label{Fig:regD}
  \end{figure}
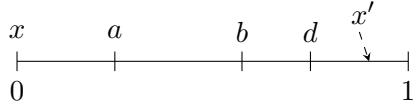


\noindent Let us by a \emph{complete} transition denote a transition of $R$ which simulates a time pass transition or a sequence of three transitions of $R$ which simulate a discrete transition. We call the states of $R$ which are not auxiliary, i.e., the states reached by complete transitions, \emph{complete} states. Figure~\ref{Fig:construction-example} shows the result of this construction applied to the timed automaton from Figure~\ref{Fig:no-sampling}.

An informal alternative description of the updates by effects is that a counter is incremented if the distance between the two corresponding clocks grows and a counter is reset to $1$ if the distance between the two corresponding clocks decreases and then the counters are updated to satisfy $D \models \overline{xyz} \Rightarrow C_{xy} + C_{yz} \leq C_{xz}$) by the $\max$ operations. We take the liberty to apply the $\max$ operations only at the end of each complete transition. This does not affect validity of Lemma~\ref{Lemma:counter-sums}, is sufficient for correctness of our construction, and it will simplify the proofs.

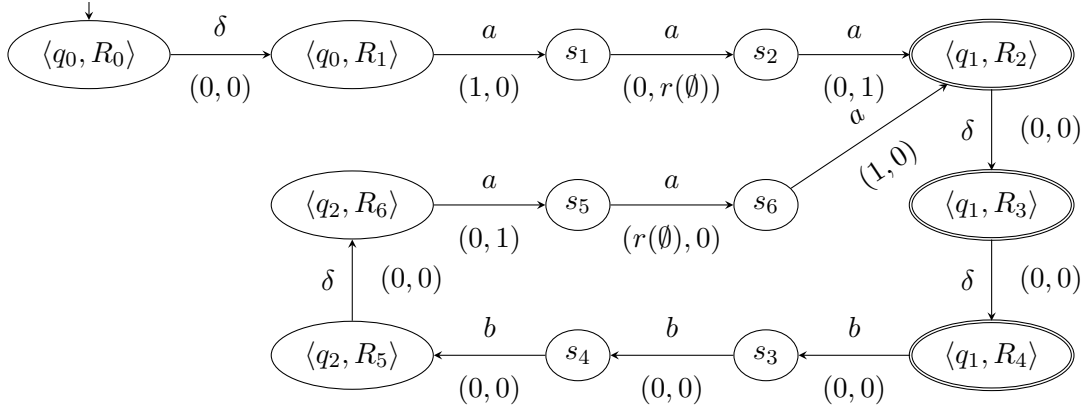
\begin{figure}[htbp]
\centering
\begin{tikzpicture}[>=stealth,every node/.style={ellipse}]

\draw(0,0) node [ellipse,draw,name=s0]{$\langle q_0, R_0 \rangle$};
\draw[xshift= 3.5cm](0,0) node [draw,name=s1]{$\langle q_0, R_1 \rangle$};
\draw[xshift= 6.5cm](0,0) node [draw,name=s2]{$s_1$};
\draw[xshift= 9cm](0,0) node [draw,name=s3]{$s_2$};
\draw[xshift= 12cm,yshift= 0cm](0,0) node [double,draw,name=s4]{$\langle q_1, R_2 \rangle$};
\draw[xshift= 12cm,yshift= -2cm](0,0) node [double,draw,name=s5]{$\langle q_1, R_3 \rangle$};
\draw[xshift= 12cm,yshift= -4cm](0,0) node [double,draw,name=s6]{$\langle q_1, R_4 \rangle$};
\draw[xshift= 9cm,yshift= -4cm](0,0) node [draw,name=s7]{$s_3$};
\draw[xshift= 6.5cm,yshift= -4cm](0,0) node [draw,name=s8]{$s_4$};
\draw[xshift= 3.5cm,yshift= -4cm](0,0) node [draw,name=s9]{$\langle q_2, R_5 \rangle$};
\draw[xshift= 3.5cm,yshift= -2cm](0,0) node [draw,name=s10]{$\langle q_2, R_6 \rangle$};
\draw[xshift= 6.5cm,yshift= -2cm](0,0) node [draw,name=s11]{$s_5$};
\draw[xshift= 9cm,yshift= -2cm](0,0) node [draw,name=s12]{$s_6$};

\draw[->] (0,0.7cm) -- (s0);

\draw[->] (s0) -- (s1) node[midway,above]{$\delta$}
                       node[midway,below]{$(0,0)$};
\draw[->] (s1) -- (s2) node[midway,above]{$a$}
                       node[midway,below]{$(1,0)$};
\draw[->] (s2) -- (s3) node[midway,above]{$a$}
                       node[midway,below]{$(0, r(\emptyset))$};
\draw[->] (s3) -- (s4) node[midway,above]{$a$}
                       node[midway,below]{$(0,1)$};
\draw[->] (s4) -- (s5) node[midway,left]{$\delta$}
                       node[midway,right]{$(0,0)$};
\draw[->] (s5) -- (s6) node[midway,left]{$\delta$}
                       node[midway,right]{$(0,0)$};
\draw[->] (s6) -- (s7) node[midway,above]{$b$}
                       node[midway,below]{$(0,0)$};
\draw[->] (s7) -- (s8) node[midway,above]{$b$}
                       node[midway,below]{$(0,0)$};
\draw[->] (s8) -- (s9) node[midway,above]{$b$}
                       node[midway,below]{$(0,0)$};
\draw[->] (s9) -- (s10) node[midway,left]{$\delta$}
                       node[midway,right]{$(0,0)$};
\draw[->] (s10) -- (s11) node[midway,above]{$a$}
                       node[midway,below]{$(0,1)$};
\draw[->] (s11) -- (s12) node[midway,above]{$a$}
                       node[midway,below]{$(r(\emptyset), 0)$};
\draw[->] (s12) -- (s4) node[midway,sloped,above]{$a$}
                       node[midway,sloped,below]{$(1,0)$};

\end{tikzpicture}

\caption{The extended R-automaton constructed for the timed automaton from Figure~\ref{Fig:no-sampling}. It has two counters $C_{xy}$ and $C_{yx}$ which are updated by effects in this order, i.e., an effect $(1,0)$ increments the counter $C_{xy}$. Complete states are labeled by a location and a region, whereas auxiliary states are labeled by $s_1, s_2, \dots$ Regions are characterized by the following constraints: $R_0: 0=x=y, R_1: 0<x=y<1, R_2: 0=x<y<1, R_3: 0<x<y<1, R_4: 0<x<y=1, R_5: 0=y<x<1, R_6: 0<y<x<1$. The automaton is not bounded, because $C_{yx}$ is incremented and never reset in the loop.}
\label{Fig:construction-example}
\end{figure}

\noindent Since we use the $\max$ operation, we need to take care of the preorder $\lesssim$. In order to do this, we need all the copy instructions in Items~\ref{Constr:regB} -- \ref{Constr:regD} and resets. Copying already assigns the desired value to the counter, which speeds up the applications of the $\max$ operation (as shown in Lemma~\ref{Lemma:counter-max-one-increase} below).

The important property of $\lesssim$ is formalized in the following lemma. The proof is rather technical and analyzes the items in the construction and the semantics of extended R-automata.

\begin{lemma}
\label{Lemma:constr-preorder}
  For all reachable complete states $\langle \langle q, D \rangle, \bar{C}, \lesssim \rangle$ of $R$, the following holds:
  \begin{enumerate}[\em(i)]
    \item \label{Lemma:constr-preorder-lnsim} $C_{bc} \lnsim C_{ad}$ if and only if for all $\nu \in D$, $\numodels{bc} < \numodels{ad}$, and

    \item \label{Lemma:constr-preorder-simeq} $C_{ab} \simeq C_{cd}$ if and only if for all $\nu \in D$, $\numodels{ab} = \numodels{cd} > 0$.
  \end{enumerate}
\end{lemma}

\proof{
  We show by induction on the length of a shortest path reaching $\langle \langle q, D \rangle$, $\bar{C}$, $\lesssim \rangle$ that the claim holds. The basic step is trivial. For the induction step, observe that the claim that for all $\nu \in D$, $\numodels{bc} < \numodels{ad}$ is equivalent to $(D \models \overline{abc} \wedge D \models \overline{cda}) \vee (\Dmodels{a}{=}{b} \wedge D \models \overline{acd}) \vee (\Dmodels{c}{=}{d} \wedge D \models \overline{abc})$ and the claim that for all $\nu \in D$. $\numodels{ab} = \numodels{cd} > 0$ is equivalent to $\Dmodels{a}{=}{c} \wedge \Dmodels{b}{=}{d}$.

  Point~(\ref{Lemma:constr-preorder-lnsim}), "$\Rightarrow$": Correctness of all inequalities introduced by Item~\ref{Constr:regA} of the construction follows from IH.

  Item~\ref{Constr:regB} of the construction introduces inequalities $C_{xu} \lnsim C_{au}$ and $C_{ux} \lnsim C_{ud}$, because of the copy instruction (Point~\ref{Pre:copy} in the semantics introduces equality) and then $C_{xu}, C_{ux}$ are incremented by the instruction $1$, which breaks the equality into inequality (Point~\ref{Pre:plus} in the semantics). But it is clear from the analysis of the region $D'$ and the observations above that the claim is satisfied. Item~\ref{Constr:regB} also introduces inequalities by resets. The reset instructions are delayed by one transition (they take place on the second transition in the sequence) and therefore the inequalities $C_{xu} \lnsim C_{au}$ and $C_{ux} \lnsim C_{ud}$ are already established. This prevents the inequalities $C_{xa} \lnsim C_{au}$, $C_{dx} \lnsim C_{ud}$ to appear in the preorder. It is easy to verify from the region that the remaining inequalities which are established satisfy the claim. It follows from IH that the inequalities introduced by the transitive closure satisfy the claim.


  Item~\ref{Constr:regC} does not introduce any new inequalities for $C_{xu}$, because there is no other counter $C_{ab}$ such that $C_{ab} \simeq C_{xu}$ (IH, Point~(\ref{Lemma:constr-preorder-simeq})). The argument for the inequalities created by copying and resets is the same as for the previous item.

  Item~\ref{Constr:regD} is dual to the previous item. 

  Point~(\ref{Lemma:constr-preorder-lnsim}), "$\Leftarrow$": The fact that all required inequalities are created by Item~\ref{Constr:regA} follows from IH.

  Items~\ref{Constr:regB} -- \ref{Constr:regD} have to create new inequalities for counters containing the clock $x$ (we can find all of them by inspecting the regions). The copy instructions put $C_{xu} \simeq C_{au}$ and $C_{ux} \simeq C_{ud}$ (Point~\ref{Pre:copy} in the semantics). The counters $C_{xu}, C_{ux}$ are then incremented by the instruction $1$, while the counters $C_{au}, C_{ud}$ stay unchanged (instruction $0$). This results in the inequalities $C_{xu} \lnsim C_{au}$ and $C_{ux} \lnsim C_{ud}$. The clocks $C_{xa}, C_{dx}$ are reset by an instruction which contains all the important clocks. This (as defined in Point~\ref{Pre:reset} of the semantics, together with the transitive closure) creates all the necessary inequalities.

  Point~(\ref{Lemma:constr-preorder-simeq}), "$\Rightarrow$": Item~\ref{Constr:regA} creates equalities by the copy instruction (Point~\ref{Pre:copy}) and the transitive closure, but the correctness follows immediately from the fact that $\Dmodels{a}{=}{x}$ and from IH (for the transitive closure).

  Items~\ref{Constr:regB} -- \ref{Constr:regD} introduce equalities by the copy instructions and the transitive closure, but because the clocks $C_{xu}, C_{ux}$ are incremented by $1$ and the clocks $C_{au}, C_{ud}$ are left unchanged, the equalities introduced by the copy instructions are broken. The equalities introduced by the transitive closure satisfy the claim (IH).

  Point~(\ref{Lemma:constr-preorder-simeq}), "$\Leftarrow$": New equalities required by the region in Item~\ref{Constr:regA} are created by the copy instructions and the transitive closure. The other required equalities follow from IH. Note that $\numodels{ax} = \numodels{xa} = 0$ for all $\nu \in D$ and therefore the equality $C_{ax} \simeq C_{xa}$ is not required.

  Items~\ref{Constr:regB} -- \ref{Constr:regD} do not move any two clocks together and therefore the claim holds from IH.\qed

}

The following two lemmas formulate the essential properties of the construction we need for the proof of the correctness of the reduction. Because of these lemmas, we do not have to refer to $\lesssim$ and $\max$ operations anymore.

\begin{lemma}
\label{Lemma:counter-transitivity}
  For all reachable complete states $\langle \langle q, D \rangle, \bar{C}, \lesssim \rangle$ of $R$, the following holds:
  \begin{enumerate}[\em(i)]
    \item \label{Lemma:counter-transitivity:i}
        if $D \models \overline{xyz}$ then $C_{xy} + C_{yz} \leq C_{xz}$,
    \item \label{Lemma:counter-transitivity:coincidence}
        if $\Dmodels{x}{=}{y}$ then $C_{xy} = 0$ and $C_{xu} = C_{yu}$, $C_{ux} = C_{uy}$ for all clocks $u$.
    \item \label{Lemma:counter-transitivity:min-distance}
        if $\Dmodels{x}{\neq}{y}$ then $C_{xy} \geq 1, C_{yx} \geq 1$.
  \end{enumerate}
\end{lemma}

\proof{
  Point~(\ref{Lemma:counter-transitivity:i}) follows directly from Lemma~\ref{Lemma:constr-preorder} and Step~\ref{Sem:max} in the definition of the semantics of extended R-automata.

  The first part of Point~(\ref{Lemma:counter-transitivity:coincidence}) follows from Item~\ref{Constr:regA} in the construction of $R$ and the fact that this counter can be changed only along a transition which leads to a state $\langle \langle q', D' \rangle, \bar{C}', \lesssim' \rangle$, where $\Dpmodels{x}{\neq}{y}$ (follows straightforwardly from the construction). The second part follows from Lemma~\ref{Lemma:constr-preorder} and an observation that counters equivalent with respect to $\lesssim$ contain the same values.

  Point~(\ref{Lemma:counter-transitivity:min-distance}) follows from a simple inductive
  argument. If $\Dmodels{x}{\neq}{y}$ holds and it did not hold in the previous state then $C_{xy}, C_{yx}$ is either updated by a copy from a counter with value greater than or equal to $1$ (Item~\ref{Constr:regA}) or by a copy or reset followed by an increment (Item~\ref{Constr:regB}). Especially, Items~\ref{Constr:regC} and~\ref{Constr:regD} cannot be applied. If $\Dmodels{x}{\neq}{y}$ holds and it held also in the previous state then $C_{xy}, C_{yx}$ is either incremented (Items~\ref{Constr:regC} and~\ref{Constr:regD}) or updated by a copy from a counter with value greater than or equal to $1$ (Item~\ref{Constr:regA}) or by a copy or reset followed by an increment (Items~\ref{Constr:regB},~\ref{Constr:regC}, and~\ref{Constr:regD}).\qed

}

The property formalized in Point~(\ref{Lemma:counter-transitivity:i}) of the previous lemma is the reason for extending the R-automata with the $\max$ operations. The preorder $\lesssim$ is a technical construction thanks to which we are able to reduce limitedness for R-automata with $\max$ operations to limitedness of R-automata.

\revision{Upper bound on max operations}

Lemma~\ref{Lemma:counter-transitivity} shows that $\max$ operations ensure a lower bound on counters. The following lemma shows that applications of the $\max$ operation do not increase the counters too much. In fact, it says that $\max$ operations can increase a counter at most by $1$ in each complete step and this only if it has not been affected by other operations. 

\begin{lemma}
\label{Lemma:counter-max-one-increase}

  Let $\langle \langle q', D' \rangle, \bar{C}', \lesssim' \rangle$ and $\langle \langle q, D \rangle, \bar{C}, \lesssim \rangle$ be two consecutive complete states in a run of $R$. Only counters $C_{uv}$ such that $\Dmodels{v}{<}{u}$ and $\Dmodels{v}{\neq}{0}$ 
  can be affected by $\max$ operation.\note{Maybe we should have the counters which can cause max here?} 
  Moreover, if $C_{uv} \neq C'_{uv}$ then $C_{uv} = C_{ux} + C_{xv} = C'_{uv} + 1$.
\end{lemma}

\proof{
    We show that the first fixed-point iteration of taking maxima satisfies this claim. Then we show by contradiction that there are no more fixed-point iterations of taking maxima. 
    %

    To show the first step, we analyze all types of transitions. For Item~\ref{Constr:regA}, the claim holds trivially. We show the claim in full detail for Item~\ref{Constr:regC}. Other items are analogical.

    The counters $C_{uv}$ such that $\Dmodels{u}{<}{v}$ and $\Dmodels{u}{\neq}{0}$ (which implies that $u, v \neq x$) are not affected by the transition. They also form a downward closed set with respect to $\lesssim$, hence they are not updated by the $\max$ operation.

    Counters $C_{xu}$ are incremented along the transition. By induction on the number of clocks with different fractional part between $x$ and $u$ in $D$ we show that these counters are not updated by the $\max$ operation. The basic step is trivial, because $C_{xu}$ is a minimal element in $\lesssim$. For the induction step, let us look at the value of the expression $\max\{C_{xu}, C_{xw} + C_{wu}\}$. From IH and the previous consideration we know that neither of the counters $C_{xw}, C_{wu}$ has been updated by the $\max$ operation in this step. Therefore, $C_{xw} = C'_{xw} + 1$ and $C_{wu} = C'_{wu}$. But since $C_{xu} = C'_{xu} + 1$ and from Lemma~\ref{Lemma:counter-transitivity} we know that $C'_{xu} \geq C'_{xw} + C'_{wu}$, we have that $\max\{C_{xu}, C_{xw} + C_{wu}\} = C_{xu}$.

    Counters $C_{ux}$ are set to $C'_{ud}$ and then incremented along the transition. By induction on the number of clocks with different fractional part between $u$ and $d$ in $D$ we show that these counters are not updated by the $\max$ operation. For the basic step, there are no clocks $C$ such that $C \lnsim C_{ud}$ and $C_{xa} = 1$. Therefore, $C_{ux} = C_{ud} + C_{dx}$. For the induction step, let us look at the value of the expression $\max\{C_{ux}, C_{uw} + C_{wx}\}$. From IH and the previous consideration we know that neither of the counters $C_{uw}, C_{wx}$ has been updated by the $\max$ operation in this step. Therefore, $C_{wx} = C'_{wd} + 1$ and $C_{uw} = C'_{uw}$. But since $C_{ux} = C'_{ud} + 1$ and from Lemma~\ref{Lemma:counter-transitivity} we know that $C'_{ud} \geq C'_{uw} + C'_{wd}$, we have that $\max\{C_{ux}, C_{uw} + C_{wx}\} = C_{ux}$.

    Therefore, the counters possibly affected by the first $\max$ application are $C_{uv}$ such that $\Dmodels{v}{<}{u}$ and $\Dmodels{v}{\neq}{0}$ (which implies that $u, v\neq x$).
    These counters are set to $C_{ux} + C_{xv}$ if $C_{ux} + C_{xv} > C'_{vu}$ (other possible candidates for the $\max$ operation have not been modified along the transition). We know from Lemma~\ref{Lemma:counter-transitivity} that $C'_{uv} \geq C'_{ud} + C'_{dx} + C'_{xv}, C_{dx} \geq 1$ and from the construction we know that $C_{ux} = C'_{ud} + 1, C_{xv} = C'_{xv} + 1$. This gives us that if $C_{uv} \neq C'_{uv}$ then $C_{uv} = C'_{uv} + 1$.

    Now we show that there are no additional iterations of the application of the $\max$ operation. Let us assume that $C_{eb}$ is updated in the second iteration by $C_{ec} + C_{cb}$. Without loss of generality, let us assume that $C_{cb}$ was updated by the $\max$ operation in the first iteration (this also means that $c,b \neq x$). Note that the set of clocks updated by the $\max$ operation in the first iteration has the clock $C_{xa}$ as a lower bound. Then we know that $C_{ec}$ was not updated by the $\max$ operation in the first iteration. Here we use the restriction on extended R-automata introduced at the beginning of this section. We also  know that $\Dmodels{b}{<}{e}, \Dmodels{e}{<}{c}$, and $\Dmodels{c}{<}{b}$. The rest of the argument applies to Item~\ref{Constr:regB}. The other items are analogical.

    We know that $C_{ec} = C'_{ec}$ (from the region and from the construction) and $C_{cb} = C'_{cb} + 1$ (from the previous argument). Also, $C'_{eb} \geq C'_{ec} + C'_{cb}$ from Lemma~\ref{Lemma:counter-transitivity}, which together with the assumption that $C_{eb}$ was updated by the $\max$ operation means that $C'_{eb} = C'_{ec} + C'_{cb}$ (it follows from the region and from the construction that $C_{eb}$ was not affected by any counter operations during this step). We know from the first iteration that $C_{cb} = C_{cx} + C_{xd} = C'_{cd} + C'_{ab} + 2$ and $C'_{cb} = C'_{cd} + C'_{da} + C'_{ab}$, where $C'_{da} = 1$. This means that $C'_{eb} = C'_{ec} + C'_{cd} + C'_{da} + C'_{ab}$. But then $C_{ex} = C'_{ed} + 1 = C'_{ec} + C'_{cd} + 1$, $C_{xb} = C'_{ab} + 1$, and $C_{eb} \geq C_{ex} + C_{xb} = C'_{eb} + 1$. Hence, $C_{eb}$ has been updated by the $\max$ operation in the first iteration and it is equal to $C_{ec} + C_{cb}$ even before the second iteration, which is a contradiction.\qed

}

The previous lemma shows that we did not need the fixed-point calculation in the definition of extended R-automata semantics. On the other hand, fixed-point calculations make these automata a more powerful tool with the same complexity of the limitedness problem as R-automata with copying (which follows from Lemma~\ref{Lemma:counter-sums}). Also, defining extended R-automata with only one fixed-point iteration would make the proof of Lemma~\ref{Lemma:constr-preorder} more complicated.

\revision{End of: Upper bound on max operations}

\paragraph{\bf Correspondence between $A'$ and $R$.} Now we formulate correspondence properties between the timed automaton $A'$ and the extended R-automaton $R$ constructed as above. Let us recall that we ignore $N$ and $M$ values of the counters and denote the $P$ values by $C$. For instance, a state $\langle \langle q, D\rangle, (\bar{N}, \bar{M}, \bar{P}), \lesssim \rangle$ is written as $\langle \langle q, D\rangle, \bar{C}, \lesssim \rangle$. Let for a state in a run of an extended R-automaton, the value of the state be the maximal counter value in this state (the $P$-value). Let for a run $\rho$ of an extended R-automaton, the maximum counter value along this run be the maximal state value along this run. This is the value $P(\rho)$, but to avoid confusion, we denote it by $\max\{\rho\}$ here.

Let us say that a valuation $\nu \in D$ satisfies the counter valuation $\bar{C}$ with the smallest step $\epsilon$ (denoted by $\epsmodels{\nu}{\bar{C}}$) if for each pair of clocks $x,y$, $\numodels{xy}/\epsilon \geq C_{xy}$ (or equivalently, $\numodels{xy} \geq C_{xy}\cdot \epsilon$).

\begin{lemma}
\label{Lemma:TA-RC-correspondence}
Let $R$ be the extended R-automaton constructed from the region graph $G$ induced by a timed automaton $A$. Let $\rho = \langle \langle q_0, \{\nu_0\}\rangle, \bar{C_0}, \emptyset \rangle \arrow{} \langle \langle q, D\rangle$, $\bar{C}$, $\lesssim \rangle$ be a run in $R$ ending in a complete state, $\sigma = \langle q_0, \{\nu_0\} \rangle \arrow{} \langle q,D \rangle$ be the corresponding path in $G$, and $\epsilon \leq 1/(4\cdot \max\{\rho\})$. For all $\nu\in D_{\epsilon}$ such that $\twoepsmodels{\nu}{\bar{C}}$ there is a run $\rho'$ in $\sampledsemantics{A}$ ending in $(q, \nu)$ such that $\rho' \models \sigma$. Also, there is a $\nu\in D_{\epsilon}$ such that $\twoepsmodels{\nu}{\bar{C}}$.
\end{lemma}

\paragraph{Remark.} This lemma requires that the valuations satisfy the counters with the smallest step $2\cdot \epsilon$. This enables us smooth time-pass transitions. If the value of a counter $C_{xy}$ is $1$ and we would allow the difference between fractional parts of $x$ and $y$ to be only $\epsilon$ then we would not be able to reach a region where $0 <_D y <_D x$ in the $\epsilon$-sampled semantics by letting the time pass. Another requirement is that each increment of a counter corresponds at most to $\epsilon / 4$ in the sampled semantics. We need this to be able to place disjoint intervals between the fractional parts of the clock values next to each other within the unit interval. In other words, we need that $1 \geq (C_{xy} + C_{yx}) \cdot 2 \epsilon$ always holds for all clocks $x$ and $y$. In the proof, we also use that $2 \geq (C_{uz} + C_{zx} + C_{xu}) \cdot 2 \epsilon$ holds for all clocks $x, u, z$ such that $\overline{uxz}$, which follows from the previous constraints.

\proof{
  By induction on the length of $\sigma$. The basic step is trivial.

  For the induction step, let us first observe that the maximum counter value along $\rho$ is greater than or equal to the maximum counter value along its prefixes. Let $\nu \in D_{\epsilon}$ and $\twoepsmodels{\nu}{\bar{C}}$. We have to find $\nu' \in D'_{\epsilon}$, $\twoepsmodels{\nu'}{\bar{C}'}$, where $\langle \langle q', D' \rangle, \bar{C}', \lesssim' \rangle$ is the  previous complete state of $\rho$, such that $\nu$ can be reached from $\nu'$ along the edge from $\langle q', D' \rangle$ to $\langle q, D \rangle$. 
  We discuss different types of this edge.

  Let us first look at the case where the edge leads to the immediate time successor. Let $x$ be a clock with the smallest fractional part in $\nu$. If $\Dmodels{0}{<}{x}$ (or, equivalently, $\fract(\nu(x)) > 0$) then $\nu'(y) = \nu(y) - \fract(\nu(x))$ for all $y\in \clocks$. If $\Dmodels{0}{=}{x}$ (equivalently, $\fract(\nu(x)) = 0$) then $\nu'(y) = \nu(y) - \epsilon$ for all $y\in \clocks$. Because the minimal distance between two clocks with different fractional parts is $2\cdot \epsilon$ (follows from IH and Lemma~\ref{Lemma:counter-transitivity}), $\nu' \in D'_\epsilon$ in both cases. Also, $\twoepsmodels{\nu'}{\bar{C'}}$, because $\bar{C'} = \bar{C}$ (instructions on all counters are $0$) and the differences between the clocks do not change. 

  We discuss an edge along which a clock (denote $x$) is reset. Then we know that $\nu'(y) = \nu(y)$ for all $y\neq x$. The case where $\Dpmodels{x}{=}{0}$ 
  clearly holds, because neither distances between the fractional parts of the clocks nor the counters change. For the other case, we discuss different types of the regions $D, D'$ corresponding to the cases in the construction of $R$ separately. Let $i$ be the integral part of the clock $x$ in $D'$. If there is a clock $z$ such that $\nupmodels{xz} = 0$ then $\nu'(x) = i + \fract(\nu(z))$.

  Otherwise, there is a clock with a different fractional part than $x$ in $D'$, because $|\clocks| \geq 2$. If there is a clock with a smaller fractional part than $x$ in $D'$ then let $b$ denote a clock with the greatest fractional part smaller than the fractional part of $x$. We place $x$ at the greatest distance from $b$ to the right enforced by some clock $z$ and the counter $C_{zx}$: $$\nu'(x) = i + \max\{\fract(\nu(z) + C'_{zx}\cdot 2 \epsilon)\ |\ \forall z\in \clocks . (z = b) \vee (\nupmodels{zb} < C'_{zx} \cdot 2 \epsilon) \}$$
  If $x$ has the smallest fractional part in $D'$ (the third case in this proof) then let $b$ denote a clock with the smallest fractional part greater than the fractional part of $x$. We place $x$ at the greatest distance from $b$ to the left enforced by some clock $z$ and the counter $C_{xz}$: $$\nu'(x) = i + \min\{\fract(\nu(z) - C'_{xz}\cdot 2 \epsilon)\ |\ \forall z \in \clocks \}$$
   Here we do not need the additional condition on clocks $z$, because they all have the fractional part greater than or equal to the fractional part of $b$. The construction of the valuation $\nu'$ for $x$ is depicted in Figure~\ref{Fig:new-value-for-x}.

  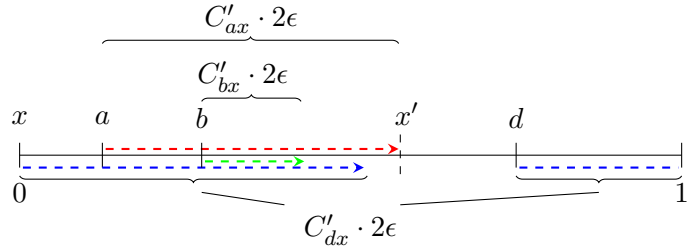
\begin{figure}[htbp]
    \centering
    \begin{tikzpicture}[>=stealth, scale=1.1]

      \path(-4,0) coordinate (0);
      \path(4,0) coordinate (1);
      \path(-4,.15) coordinate (0u);
      \path(-4,-.15) coordinate (0l);
        \path(-3.96,-.15) coordinate (0lr);
        \path(.16,-.15) coordinate (end-d);
      \path(4,.15) coordinate (1u);
      \path(4,-.15) coordinate (1l);
        \path(3.96,-.15) coordinate (1ll);
      \path(-3,.15) coordinate (bu);
        \path(-2.96,.075) coordinate (bur);
      \path(-3,-.15) coordinate (bl);
        \path(-3.1,-.075) coordinate (bll);
      \path(.6,.25) coordinate (xu);
      \path(.6,-.25) coordinate (xl);
      \path(-1.8,.15) coordinate (cu);
      \path(-1.8,-.15) coordinate (cl);
        \path(-1.76,-.075) coordinate (clr);
        \path(-.56,-.075) coordinate (end-c);
      \path(2,.15) coordinate (du);
        \path(2.04,-.15) coordinate (dlr);
        \path(.58,.075) coordinate (new-x);
      \path(2,-.15) coordinate (dl);

      \draw(-4,-.45) node [name=zero]{$0$};
      \draw(4,-.45) node [name=one]{$1$};
      \draw(-4,.45) node [name=x]{$x$};
      \draw(-3,.45) node [name=b]{$a$};
      \draw(.68,.5) node [name=b]{$x'$};
      \draw(-1.8,.45) node [name=c]{$b$};
      \draw(2,.45) node [name=d]{$d$};


      \draw (0)--(1);
      \draw (0u)--(0l);
      \draw (1u)--(1l);
      \draw (bu)--(bl);
      \draw[dashed] (xu)--(xl);
      \draw[->,dashed,thick,red] (bur)--(new-x);
      \draw[->,dashed,thick,green] (clr)--(end-c);
      \draw[dashed,thick,blue] (dlr)--(1ll);
      \draw[->,dashed,thick,blue] (0lr)--(end-d);
      \draw (cu)--(cl);
      \draw (du)--(dl);

      \draw[snake=brace] (-3, 1.35) -- (0.58,1.35) node at (-1.2,1.65)       {$C'_{ax}\cdot 2 \epsilon$};
      \draw[snake=brace] (-1.8, 0.65) -- (-0.6,0.65) node at (-1.3,0.95)      {$C'_{bx}\cdot 2 \epsilon$};

      \draw[snake=brace] (0.2,-0.25) -- (-4, -0.25);
      \draw[snake=brace] (4,-0.25) -- (2, -0.25) node at (0,-.9) {$C'_{dx}\cdot 2 \epsilon$};
      \draw (-1.8,-0.45) -- (-.8,-.6);
      \draw (3,-0.45) -- (.6,-.6);

    \end{tikzpicture}
    \caption{Illustration of the calculation of the value of $x$ in the valuation $\nu'$.     The positions of the clocks correspond to the valuation $\nu$, where $\nu(x) = 0$ ($x$ was reset), $\nu(a) = 0.12, \nu(b) = 0.24, \nu(d) = 0.76$. The values of the counters are $C_{ax} = 10, C_{bx} = 3, C_{dx} = 17$. The sampling rate is $\epsilon = 0.02$ and thus $C_{ax}\cdot\epsilon = 0.4, C_{bx}\cdot\epsilon = 0.12, C_{ax}\cdot\epsilon = 0.68$. Then $\nu'(x)$, depicted by $x'$, is $\max\{0.52, 0.36, 0.44\} = 0.52$.}
    \label{Fig:new-value-for-x}
  \end{figure}

  \noindent As the first case we consider regions $D'$ which have a clock $a$ with zero fractional part (Item~\ref{Constr:regA} in the construction, depicted in Figure~\ref{Fig:regA}). We denote a clock with the greatest fractional part smaller than the fractional part of $x$ by $b$ (there is always one such clock, since $b$ could be the clock $a$). If it exists, then we also denote a clock with the smallest fractional part greater than the fractional part of $x$ by $c$.

\forget{
%
%

  \begin{figure}[htbp]
    \caption{The region $D'$ has a clock $a$ with zero fractional part.}
    \label{Fig:regA}
  \end{figure}
}
  We have to show that $\nu'\in D'_\epsilon$ and that $\twoepsmodels{\nu'}{\bar{C'}}$. First we show that $\nu'\in D'_\epsilon$. If there is a clock $y$ such that $\Dpmodels{x}{=}{y}$ then clearly $\nu'\in D'_\epsilon$. Otherwise, we have to show that $\Dpmodels{b}{<}{x}$ 
  and if $c$ exists then also that $\Dpmodels{x}{<}{c}$. 
  To show that $\Dpmodels{b}{<}{x}$, 
  we need to show that $\fract(\nu(b)) + C'_{bx} \cdot 2 \epsilon < 1$ and then the rest follows from the construction of $\nu'$. Since $C_{ba} = C'_{ba} \geq C'_{bx} + C'_{xa}$ and $C'_{xa} \geq 1$ (Lemma~\ref{Lemma:counter-transitivity}), we have that $C_{ba} > C'_{bx}$ and from the fact that $\twoepsmodels{\nu}{\bar{C}}$ we have that $\numodels{ba} > C'_{bx}\cdot 2 \epsilon$ and thus $\fract(\nu(b)) + C'_{bx}\cdot 2 \epsilon < 1$.  To show that $\Dpmodels{x}{<}{c}$, 
  we discuss the following two cases. Let us denote the clock chosen by the $\max$ function in the construction of the value $\nu'(x)$ by $z$.

  \begin{enumerate}[$\bullet$]
    \item If the clock $z$ has the same fractional part as $c$ in $D'$ then the claim follows from the condition $\nupmodels{zb} < C'_{zx}\cdot 2 \epsilon$ in the construction of $\nu'$ and the observation that $C'_{cx}\cdot 2 \epsilon < 1$.
    \item Otherwise, we have that $C_{zc} = C'_{zc} \geq C'_{zx} + C'_{xc}$ and $C'_{xc} \geq 1$ (Lemma~\ref{Lemma:counter-transitivity}), thus $C_{zc} > C'_{zx}$. From the fact that $\twoepsmodels{\nu}{\bar{C}}$ and from the construction of $\nu'$ we have that $\nupmodels{zc} \geq C'_{zc}\cdot 2 \epsilon$ and $\nupmodels{zx} = C'_{zx}\cdot 2 \epsilon$, which gives that $\nupmodels{zc} > \nupmodels{zx}$. This is a sufficient condition in case that $\Dpmodels{z}{<}{c}$. Otherwise, we need to show that $\Dpmodels{x}{<}{z}$, which is shown in the previous item.
  \end{enumerate}

\noindent Now we show that $\twoepsmodels{\nu'}{\bar{C'}}$. If there is a clock $y$ such that $\Dpmodels{x}{=}{y}$ then the fact that $\twoepsmodels{\nu'}{\bar{C'}}$ follows directly from Lemma~\ref{Lemma:counter-transitivity}. Otherwise, we have to check all the counters. For all counters $C'_{uv}$ such that $u, v \neq x$, $C'_{uv} = C_{uv}$ and from the construction of $\nu'$, $\nupmodels{uv} \geq C'_{uv}\cdot 2 \epsilon$. For counters $C'_{ux}$ (for all clocks $u$), the fact that $\nupmodels{ux} \geq C'_{ux}\cdot 2 \epsilon$ follows directly from the construction of $\nu'$ (and from the fact that $\Dpmodels{b}{<}{x}$ for the clocks which do not satisfy the condition in the construction of $\nu'$). For the counters $C'_{xu}$ we consider two cases. Let us denote the clock chosen by the $\max$ function in the construction of the value $\nu'(x)$ by $z$.

  \begin{enumerate}[$\bullet$]
    \item If the clock $z$ does not have the same fractional part as $u$ in $D'$ then we have again two possibilities.

        \begin{enumerate}[$-$]
          \item If $D' \models \overline{xuz}$ then we have that $C_{zu} = C'_{zu} \geq C'_{zx} + C'_{xu}$ (Lemma~\ref{Lemma:counter-transitivity}). From the fact that $\twoepsmodels{\nu}{\bar{C}}$ and from the construction of $\nu'$ we have that $\nupmodels{zu} \geq C'_{zu}\cdot 2 \epsilon$ and $\nupmodels{zx} = C'_{zx}\cdot 2 \epsilon$,  therefore $\nupmodels{xu} = \nupmodels{zu} - \nupmodels{zx} > C'_{xu}\cdot 2 \epsilon$.
          \item If $D' \models \overline{xzu}$ then we have that $\nupmodels{xu} > \nupmodels{xz}$. From the construction of $\nu'$ we have that $\nupmodels{xz} = 1 - (C'_{zx}\cdot 2 \epsilon)$ and from the condition on $\epsilon$ we have that $1 \geq 2\cdot \max\{C'_{zx}, C'_{xu}\}\cdot 2\epsilon$. This together gives that $\nupmodels{xz} \geq (C'_{xu}\cdot 2 \epsilon)$.
        \end{enumerate}

    \item If the clock $z$ has the same fractional part as $u$ in $D'$ then it suffices to observe that $C'_{xu} + C'_{ux} \leq 2\cdot \max\{\rho\}$ and thus $1 \geq (C'_{xu} + C'_{ux})\cdot 2\epsilon$. From the construction of $\nu'$ we have that $\nupmodels{xu} = 1 - (C'_{ux}\cdot 2 \epsilon)$ and thus $\nupmodels{xu} \geq C'_{xu} \cdot 2 \epsilon$.
  \end{enumerate}

\noindent As the second case we consider regions $D'$ such that $0 <_{D'} \Dpmodels{a}{<}{x} <_{D'} d$ (Item~\ref{Constr:regA} in the construction, depicted in Figure~\ref{Fig:regA}).
  %
  %
  The argument for this case is the same as for the first case, with the only difference that we use the counters $C_{bd}, C_{xd}$ instead of the counters $C_{ba}, C_{xa}$ when showing that $\Dpmodels{b}{<}{x}$.

  As the third case we consider regions $D'$ where $x$ has strictly smaller fractional part than other clocks (Item~\ref{Constr:regC} in the construction, depicted in Figure~\ref{Fig:regC}). We denote a clock with the smallest fractional part greater than the fractional part of $x$ by $a$ (there is always one such clock, since $|\clocks|\geq 2$). 

\forget{
  \begin{figure}[htbp]
    \caption{The clock $x$ has the smallest (strictly) fractional part in the region $D$.}
    \label{Fig:regC}
  \end{figure}
}


  We have to show that $\nu'\in D'_\epsilon$ and that $\twoepsmodels{\nu'}{\bar{C'}}$. First we show that $\nu'\in D'_\epsilon$. We have to show that $(\fract(\nu(z)) - C'_{xz}\cdot 2 \epsilon) > 0$ for all clocks $z$ and that $\Dpmodels{x}{<}{a}$. 
  The first part follows from the fact that $\numodels{xz} \geq C_{xz}\cdot 2 \epsilon$, $\nu(z) = \nu'(z)$, and $C'_{xz} < C_{xz}$. At this place, we use the fact that the value of $C_{xz}$ is incremented along these transitions in the extended R-automaton construction. The second fact follows from the first one and from the fact that $C'_{xa} \geq 1$ (Lemma~\ref{Lemma:counter-transitivity}).

  Now we show that $\twoepsmodels{\nu'}{\bar{C'}}$. The argument is 'dual' to the argument for the first case. For all counters $C'_{uv}$ such that $u, v \neq x$, $C'_{uv} \leq C_{uv}$ and from the construction of $\nu'$, $\nupmodels{uv} \geq C'_{uv}\cdot 2 \epsilon$. For counters $C'_{xu}$ (for all clocks $u$), the fact that $\nupmodels{xu} \geq C'_{xu}\cdot 2 \epsilon$ follows directly from the construction of $\nu'$. For the counters $C'_{ux}$ we consider two cases. Let us denote the clock chosen by the $\min$ function in the construction of the value $\nu'(x)$ by $z$.

  \begin{enumerate}[(1)]
    \item If the clock $z$ does not have the same fractional part as $u$ in $D'$ then we have again two possibilities.

        \begin{enumerate}[(a)]
          \item If $D' \models \overline{xzu}$ then we have that $C_{uz} = C'_{uz} \geq C'_{ux} + C'_{xz}$ (Lemma~\ref{Lemma:counter-transitivity}). From the fact that $\twoepsmodels{\nu}{\bar{C}}$ and from the construction of $\nu'$ we have that $\nupmodels{uz} \geq C'_{uz}\cdot 2 \epsilon$ and $\nupmodels{xz} = C'_{xz}\cdot 2 \epsilon$, therefore $\nupmodels{ux} = \nupmodels{uz} - \nupmodels{xz} > C'_{ux}\cdot 2 \epsilon$.
          \item If $D' \models \overline{xuz}$ then we have that $\nupmodels{ux} > \nupmodels{zx}$. From the construction of $\nu'$ we have that $\nupmodels{zx} = 1 - (C'_{xz}\cdot 2 \epsilon)$ and from the condition on $\epsilon$ we have that $1 \geq 2\cdot \max\{C'_{xz}, C'_{ux}\}\cdot 2 \epsilon$. This together gives that $\nupmodels{zx} \geq (C'_{ux}\cdot 2 \epsilon)$.
        \end{enumerate}

    \item If the clock $z$ has the same fractional part as $u$ in $D'$ then it suffices to observe that $C'_{xu} + C'_{ux} \leq 2\cdot \max\{\rho\}$ and thus $1 \geq (C'_{xu} + C'_{ux})\cdot 2 \epsilon$. From the construction of $\nu'$ we have that $\nupmodels{ux} = 1 - (C'_{xu}\cdot 2 \epsilon)$ and thus $\nupmodels{ux} \geq C'_{ux} \cdot 2 \epsilon$.
  \end{enumerate}

\noindent As the fourth case we consider regions $D'$ where $x$ has strictly greater fractional part than other clocks (Item~\ref{Constr:regC} in the construction, depicted in Figure~\ref{Fig:regC}). We denote a clock with the greatest fractional part smaller than the fractional part of $x$ in $D'$ by $d$ (there is always one such clock, since $|\clocks|\geq 2$). The correctness argument is 'dual' to the argument from the third case.

  We have to show that $\nu'\in D'_\epsilon$ and that $\twoepsmodels{\nu'}{\bar{C'}}$. First we show that $\nu'\in D'_\epsilon$. We have to show that $(\fract(\nu(z)) + C'_{zx}\cdot 2 \epsilon) < 1$ for all clocks $z$ and that $\Dpmodels{d}{<}{x}$. 
  The first part follows from the fact that $\numodels{zx} \geq C_{zx} \cdot 2 \epsilon$, $\nu(z) = \nu'(z)$, and $C'_{zx} < C_{zx}$. At this place, we use the fact that the value of $C_{zx}$ is incremented along these transitions in the extended R-automaton construction. The second fact follows from the first one and from the fact that $C'_{dx} \geq 1$ (Lemma~\ref{Lemma:counter-transitivity}).

  Now we have to show that $\twoepsmodels{\nu'}{\bar{C'}}$. The argument is the same as the argument for the first case, with the difference that for the counters $C'_{uv}$ such that $u, v \neq x$, we have that $C'_{uv} \leq C_{uv}$.

%
%
%
%

  It remains to show that there is a valuation $\nu \in D_\epsilon$ such that $\twoepsmodels{\nu}{\bar{C}}$. We construct $\nu$ in the following way. Let the integral parts of all clocks correspond to $D$.
  Let $a$ be a clock with the smallest fractional part in $D$. If $\Dmodels{a}{=}{0}$ then $\fract(\nu(a)) = 0$, otherwise, $\fract(\nu(a)) = \epsilon$. For all other clocks $b$, let $\fract(\nu(b)) = C_{ab} \cdot 2 \epsilon$. Correctness of this assignment (for all $b$, $C_{ab} \cdot 2 \epsilon < 1 - \epsilon$) follows from the condition on $\epsilon$.
  %

  We need to show that $\nu \in D_\epsilon$ and that $\twoepsmodels{\nu}{\bar{C}}$. The former follows directly from Lemma~\ref{Lemma:counter-transitivity} and the latter from the following consideration. For all clocks $\Dmodels{c}{\leq}{d}$, $\numodels{cd} \geq C_{cd}\cdot 2 \epsilon$, because of the fact that $C_{ac} + C_{cd} \leq C_{ad}$ (Lemma~\ref{Lemma:counter-transitivity}) and $\numodels{dc} \geq C_{dc}\cdot 2 \epsilon$, because $\numodels{dc} > 1 - (C_{ad}\cdot 2 \epsilon)$ and $1 \geq (C_{ad} + C_{dc})\cdot 2 \epsilon$. We also know that $\numodels{ba} \geq C_{ba}\cdot 2 \epsilon$ for all clocks $b$, because $\epsilon \leq 1/(4\cdot \max\{\rho\})$.\qed

}

We also prove that the maximum counter value of a path constrains $\epsilon$ from above.

%
%

\begin{lemma}
\label{Lemma:cost-is-real}
Let $R$ be the extended R-automaton constructed from the region graph $G$ induced by a timed automaton $A$. Let $\rho = \langle \langle q_0, \{\nu_0\}\rangle, \bar{C_0}, \emptyset \rangle \arrow{} \langle \langle q, D\rangle$, $\bar{C}$, $\lesssim \rangle$ be a run in $R$, $\sigma = \langle q_0, \{\nu_0\} \rangle \arrow{} \langle q,D \rangle$ be the corresponding path in $G$ and $\rho' = \langle q_0, \nu_0 \rangle \arrow{}_\epsilon \langle q,\nu \rangle$ be a run in $\sampledsemantics{A}$ for some $\epsilon$ such that $\rho' \models \sigma$. Then for all pairs of clocks $x,y$, $\numodels{xy} \geq C_{xy} \cdot \epsilon$.
\end{lemma}

\proof{
  By induction on the length of $\sigma$. The basic step is trivial. For the induction step, we show that if the runs of $R$ and $A$ end in the states $\langle \langle q', D'\rangle, \bar{C'}, \lesssim' \rangle$ and $\langle q',\nu' \rangle$, respectively, satisfying the condition, i.e., for all pairs of clocks $x,y$, $\numodels{xy} \geq C_{xy} \cdot \epsilon$, then the condition is also satisfied after transitions leading to the next (complete) states $\langle \langle q, D\rangle, \bar{C}, \lesssim \rangle$ and $\langle q,\nu \rangle$. We discuss the types of transitions.

  We first discuss the case where the edge leads to the immediate time successor. The condition is clearly satisfied, because neither the differences between the clocks nor the counter values change after a time transition.

  We discuss an edge along which a clock (denote $x$) is reset. The case where $\fract(\nu'(x)) = 0$ ($x$ has zero fractional part in $D'$, $\Dpmodels{x}{=}{0}$) clearly keeps the condition satisfied, because neither the differences between the clocks nor the counter values change after reset of $x$. For the other case, we discuss several different types of the regions $D, D'$.

  As the first case we consider the situation where the region $D'$ has a clock $a$ with zero fractional part (depicted in Figure~\ref{Fig:regA}, Item~\ref{Constr:regA} in the construction). For the clocks $u,v$ different from the clock $x$, the distances between the fractional parts do not change and $C_{uv} = C'_{uv}, C_{vu} = C'_{vu}$. For each clock $u$, $C_{xu} = C'_{au}$, $C_{ux} = C'_{ua}$, hence the condition is satisfied from IH.

  As the second case we consider the situation where Item~\ref{Constr:regB} in the construction applies. There, the region $D'$ has clocks $a, d$ such that the fractional part of the clock $a$ is smaller than or equal to the fractional part of $x$ and  the fractional part of the clock $d$ is greater than or equal to the fractional part of $x$ (depicted in Figure~\ref{Fig:regB}). We denote a clock with the greatest fractional part smaller than the fractional part of $x$ by $b$ (there is always one such clock, since $b$ could be the clock $a$). We also denote a clock with the smallest fractional part greater than or equal to the fractional part of $x$ by $c$ (there is always one such clock, since $b$ could be the clock $d$).

  First, we look at the distances $\overline{xa}, \overline{dx}, \overline{da}$. We have that $C_{xa} = C_{dx} = 1$, but already from the region we know that $\numodels{xa} \geq \epsilon, \numodels{dx} \geq \epsilon$. Lemma~\ref{Lemma:counter-transitivity} gives us that $C_{da} = \max\{C'_{da}, C_{xa} + C_{dx} = 2\}$, so the condition either holds from IH ($C'_{da} > 2$) or because $\numodels{xa} \geq 2\cdot \epsilon$ (from the region).

  For the distances between the clocks $a$ and $d$ (avoiding $x$ in $D$), neither distances nor the counter values change.

  For the distances between the clocks $c$ and $b$ different from $x$ such that $\Dmodels{b}{<}{c}$ (alternatively, $\fract(\nu(b)) < \fract(\nu(c))$) 
  we have to analyze the counters carefully. (This is the case where we pass through $x$ in $D$ when going from $c$ to $b$; in the following argumentation we assume that $c$ is different from $a$ and $b$ is different from $d$, but it is easy to see that the same arguments, even a bit simplified, would work if this assumption does not hold.) If $C_{cb} = C'_{cb}$ then the validity of the condition holds from IH. If $C_{cb} > C'_{cb}$ then from Lemma~\ref{Lemma:counter-max-one-increase} we know that $C_{cb} = C_{cx} + C_{xb} = C'_{cb} + 1$. From the construction, $C_{cx} = C'_{cd} + 1$ and $C_{xb} = C'_{ab}  + 1$. From Lemma~\ref{Lemma:counter-transitivity} we have that $C'_{da} \geq 1$ and $C'_{cb} \geq C'_{cd} + C'_{da} + C'_{ab}$. Then $C'_{cb} = C'_{cd} + C'_{da} + C'_{ab}$ and $C'_{da} = 1$. From this it follows that $C_{cb} = C'_{cd} + 2 + C'_{ab}$.
  %
  %
  We also have that $\numodels{cb} = \numodels{cd} + \numodels{da} + \numodels{ab}$. From IH we know that $\numodels{cd} = \nupmodels{cd} \geq C'_{cd}\cdot\epsilon$, $\numodels{ab} = \nupmodels{ab} \geq C'_{ab}\cdot\epsilon$, from the region we have that $\numodels{da} \geq 2\cdot\epsilon$. Together, $\numodels{cb} \geq C_{cb}\cdot \epsilon$.

  Now we look at the distances between $x$ and other clocks in the region denoted  $b$ such that $\Dmodels{a}{<}{b}$. Directly from the construction of $R$ we have that $C_{xb} = 1 + C_{ab} = 1 + C'_{ab}$. From IH we know that $\numodels{ab} = \nupmodels{ab} \geq C'_{ab}\cdot\epsilon$, from the region we have that $\numodels{xa} \geq \epsilon$. Since $\numodels{xb} = \numodels{xa} + \numodels{ab}$, all together gives that $\numodels{xb} \geq C_{xb}\cdot \epsilon$.

%
%

  It remains to check the distances between clocks $b$ such that $\Dmodels{b}{<}{d}$ and $x$. This case is symmetrical to the previous case.

  As the third case we consider the situation where Item~\ref{Constr:regC} in the construction applies. There, $x$ has strictly smaller fractional part than other clocks in the region (depicted in Figure~\ref{Fig:regC}). We denote a clock with the smallest fractional part greater than the fractional part of $x$ by $a$ (there is always one such clock, since $|\clocks|\geq 2$). We also denote a clock with the greatest fractional part in $D'$ by $d$ (there is always one such clock, since it could also be $a$).

  First, we look at the distances between $x$ and other clocks in the region, denoted $b$. From the construction we have that $C_{xb} = C'_{xb} + 1$. From IH we know that $\nupmodels{xb} \geq C'_{xb} \cdot\epsilon$, from the region we have that $\numodels{xb} \geq \nupmodels{xb}+\epsilon$. Together, $\numodels{xb} \geq C_{xb}\cdot \epsilon$.

  Now we check the distances between clocks in the region denoted by $b$ and $x$. $C_{dx} = 1$, but already from the region we know that $\numodels{dx} \geq \epsilon$. For the other clocks we have directly from the construction of $R$ that $C_{bx} = C_{bd} + 1 = C'_{bd} + 1$. From IH we know that $\numodels{bd} = \nupmodels{bd} \geq C_{bd}\cdot\epsilon$, from the region we have that $\numodels{da} \geq \epsilon$. Since $\numodels{bx} = \numodels{bd} + \numodels{dx}$, all together gives that $\numodels{bx} \geq C_{bx}\cdot \epsilon$.

  For the distances between the clocks $c$ and $b$ different from $x$ such that $\Dmodels{b}{<}{c}$ (alternatively, $\fract(\nu(b)) < \fract(\nu(c))$)
  we have to analyze the counters carefully. (In the following argumentation we assume that $c$ is different from $a$ and $b$ is different from $d$, but it is easy to see that the same arguments, even a bit simplified, would work if this assumption does not hold.) If $C_{cb} = C'_{cb}$ then the validity of the condition holds from IH. If $C_{cb} > C'_{cb}$ then we know from Lemma~\ref{Lemma:counter-max-one-increase} that $C_{cb} = C_{cx} + C_{xb} = C'_{cb} + 1$. From the construction, $C_{cx} = C'_{cd} + 1$ and $C_{xb} = C'_{xb} + 1$. From Lemma~\ref{Lemma:counter-transitivity} we have that $C'_{cb} \geq C'_{cd} + C'_{xb} + C'_{dx}$ and that $C'_{dx} \geq 1$. Then $C'_{cb} = C'_{cd} + C'_{dx} + C'_{xb}$ and $C'_{dx} = 1$. From this it follows that $C_{cb} = C'_{cd} + 2 + C'_{xb}$.
  We also have that $\numodels{cb} = \numodels{cd} + \numodels{dx} + \numodels{xb}$. From IH we know that $\numodels{cd} = \nupmodels{cd} \geq C'_{cd}\cdot\epsilon$, we have shown that $\numodels{xb} \geq (C'_{xb}+1)\cdot\epsilon$, from the region we have that $\numodels{dx} \geq \epsilon$. Together, $\numodels{cb} \geq C_{cb}\cdot \epsilon$.

  For the distances between the clocks $a$ and $d$ (avoiding $x$ in $D, D'$), neither distances nor the counter values change.

  As the fourth case we consider the situation where Item~\ref{Constr:regD} in the construction applies. This case is dual to the third case.\qed

}

\section{Decidability Proof}
\label{Sec:proof}

First we show that Theorem~\ref{Thm:Decidability} is true for timed automata with one clock.

\begin{lemma}
\label{Lemma:one-clock}
  For a given timed automaton $A$ with the set of clocks $\clocks$ such that $|\clocks| = 1$, $L_{1/2}(A) = L(A)$ and $L_{1/2}^\omega(A) = L^\omega(A)$.
\end{lemma}

\proof{
  Let us denote the clock by $x$. For each run over $w$ in $\densesemantics{A}$, we construct a run in $\semantics{A}{1/2}$ as follows. We modify the time delays so that all discrete transitions taken with $\intg(x) = i$ and $\fract(x) \neq 0$ are now taken with $\intg(x) = i$ and $\fract(x)=1/2$. Clearly, there is such a run in $\densesemantics{A}$, because for all $i \in \Nat$, all valuations with $\intg(x) = i$ and $\fract(x) \neq 0$ are untimed bisimilar. Such a run is also a run in $\semantics{A}{1/2}$.\qed
}

For the other cases, we first show how to transform a given timed automaton into a timed automaton which resets at most one clock along each transition and which is equivalent with respect to the sampling problem. For each discrete transition labeled by $a$ with a guard $g$ and reset $Y\subseteq \clocks$, we create a sequence of $|\clocks|$ transitions (and $|\clocks| -1$ auxiliary non-accepting states between them) labeled by $a$. These transitions reset clocks from $Y$ one by one. If $Y\neq \emptyset$ then let us denote the first reset clock by $x$. The first transition is guarded by $g$ and the guards on the other transitions are either $g$ if $|Y|=\emptyset$ or $x=0$ otherwise.
%
%

\begin{lemma}
\label{Lemma:removing-multiple-resets}
  For a given timed automaton $A$ with the set of clocks $\clocks$, the timed automaton $A'$ with at most one reset along each transition constructed as above is equivalent to $A$ with respect to the sampling problem.
\end{lemma}

\proof{
  Let $h(\Sigma^* \arrow{} \Sigma^*)$ be a homomorphism with respect to the word concatenation defined by $h(a) = a^{|\clocks|}, a\in \Sigma$. Clearly, $w\in L(A)$ if and only if $h(w) \in L(A')$. For a run $\rho$ over $w$ in $\densesemantics{A}$, we can construct a run $\rho$ over $h(w)$ in $\densesemantics{A'}$ using the same time delays as $\rho$ by taking no delays in the auxiliary states. For a run $\rho$ over $h(w)$ in $\densesemantics{A'}$, we can construct a run $\rho$ over $w$ in $\densesemantics{A}$ using the delays which are sums of the time delays from $\rho$ by adding up all delays from the auxiliary states. Observe that when at least one clock is reset along a transition in $A$ then the delays in the corresponding auxiliary states are zero.\qed
}

The next lemma shows how to remove the transitions labeled by $\delta$ in the extended R-automaton $R$ constructed in Section~\ref{Sec:Encoding-TA-RC}. We use the same algorithm as is used for removing $\epsilon$-transitions in finite automata. Each sequence of transitions $s_1 \arrow{\delta, (0, \dots, 0)} \dots \arrow{\delta, (0, \dots, 0)} s_{k-1} \arrow{a, (e_1, \dots, e_n)} s_k$ is replaced by the transition $s_1 \arrow{a, (e_1, \dots, e_n)} s_k$. Clearly, this construction results in an extended R-automaton. Let for a word $w$ and an extended R-automaton $R$, $c_R(w) = \min\{B | w \in L_B(R)\}$ (where $\min\{\} = \omega$). Let $w\upharpoonright \delta$ for $w\in (\Sigma \cup \{\delta\})^*$ denote the projection of $w$ to $\Sigma^*$ (we skip all letters $\delta$). Let $w \bigvee w'$ denote shuffle of the two words.

\begin{lemma}
\label{Lemma:removing-t-transitions}
  Let $R$ be an extended R-automaton constructed in Section~\ref{Sec:Encoding-TA-RC} and $R'$ be the extended R-automaton constructed as above. Then for each $w \in L(R')$ there is $k\in\Nat$ and $w' \in w \bigvee \delta^k$ such that $w' \in L(R)$ and $c_{R'}(w) = c_R(w')$. Also, for each $w\in L(R)$, $w\upharpoonright \delta \in L(R')$ and $c_{R'}(w\upharpoonright \delta) \leq c_R(w)$.
\end{lemma}

\proof{
  The proof follows directly from the fact that the effect $(0, \dots, 0)$ does not change the counter values and the preorder $\lesssim$.\qed
}

Let $h$ be a homomorphism which triples each letter in the word, i.e., $h(a) = aaa$ for all $a\in\Sigma$. Now we have all tools to prove the main theorem.

\proof[Proof of Theorem~\ref{Thm:Decidability}]{

First, we show the claim for finite words -- decidability of the sampling problem.
Lemma~\ref{Lemma:removing-multiple-resets} allows us to consider only timed automata with at most one reset along each transition. For such a timed automaton $A$, we construct an extended R-automaton $R$ as described in Section~\ref{Sec:Encoding-TA-RC} and an extended R-automaton $R'$ as described above. According to Lemma~\ref{Lemma:RC-automata}, it is decidable whether the language of an extended R-automaton is limited.

If the language of $R'$ is limited by a natural number $B$ then let $\epsilon = 1 / (4 \cdot B)$. For each (untimed) word $w\in L(A)$ there is a run of $R'$ which accepts $h(w)$ with counters bounded by $B$. From Lemma~\ref{Lemma:removing-t-transitions} we know that there is a number $k$ and a word $w' \in h(w) \bigvee \delta^k$ such that $w'$ is accepted by $R$ with counters bounded by $B$. Having an accepting run of $R$, Lemma~\ref{Lemma:TA-RC-correspondence} says that $A$ accepts $w$ in $\epsilon$-sampled semantics (it accepts a timed word whose untimed version is $w$).

Assume that the language of $R'$ is not limited. For each $\epsilon = 1/B$ where $B$ is a natural number we find a word $h(w)$ such that some counter $C_{xy}$ exceeds $B$ along each accepting run of $R'$. From Lemma~\ref{Lemma:removing-t-transitions} we know that for all $k$ and for all $w' \in h(w) \bigvee \delta^k$, there is no accepting run of $R$ over $w'$ with counters bounded by $B$. According to Lemma~\ref{Lemma:cost-is-real}, there is no accepting run of $A$ over $w$ in $\epsilon$-sampled semantics, because it would have to visit a state $(q, \nu)$ with $\numodels{xy} > B\cdot \epsilon = 1$.


The following shows decidability of the $\omega$-sampling problem.
The $\omega$-limitedness problem is decidable for extended R-automata over $\omega$-words with B\"{u}chi acceptance conditions. It has been show that $\omega$-universality is decidable for R-automata in~\cite{aky08r-automata}.  In the same way as for the finite words case, we can use this result to show that $\omega$-limitedness is decidable for R-automata. Then the decidability of $\omega$-limitedness for extended R-automata follows from Lemmas~\ref{Lemma:decidability-of-copy},~\ref{Lemma:counter-maxs}, and~\ref{Lemma:counter-sums}.



If the extended R-automaton $R$ constructed from a given timed automaton $A$ is $\omega$-limited then we show that $A$ can be $\omega$-sampled as follows. From the finite word case we know that there is an $\epsilon$ such that each (finite) prefix of a $w \in L^\omega(R)$ has a corresponding concrete run of $A$ in $\epsilon$-sampled semantics.
We show how to construct all prefixes of an infinite accepting concrete run of $A$ over $w$ in $\epsilon$-sampled semantics.
The basic idea behind this construction is that each $\epsilon$ gives us an equivalence  relation on valuations with finite index (defined formally below). 
This means that there are only finitely many possible transitions from each state. Therefore, we have an infinite tree induced by the runs over prefixes which is finitely branching. According to K\"onig's Lemma, this tree has an infinite branch.


Now we formalize the previous intuition. Let $B$ be a natural number such that $L_B^\omega(R) = L^\omega(R)$. Let $\rho$ be an accepting run over $w \in \Sigma^\omega$ with $\max\{\rho\} \leq B$ and let $\epsilon = 1/(4\cdot B)$. Let us denote by $H$ the set of concrete runs of $A$ along all prefixes of $\rho$ given by Lemma~\ref{Lemma:TA-RC-correspondence}.

First, we define an equivalence relation $\sim_K$ on clock valuations by $\nu \sim_K \nu'$ if for all clocks $x$, $\nu(x) \neq \nu'(x)$ implies $\nu(x) > K$ and $\nu'(x) > K$. Let $K$ be the greatest constant which appears in $A$. It is easy to see that for each $\epsilon$, $\sim_K$ has a finite index on the set of valuations $\{\nu | \forall x\in\clocks \exists k \in\Nat. \nu(x) = k\cdot \epsilon \}$. Also, $\sim_K \subseteq \cong_K$.

We construct the prefixes inductively. We assume that we can build a prefix of length $j$ ending in a state $\langle q, \nu \rangle$ such that there is an infinite subset of $H$ containing only runs whose $j$-th state is $\langle q, \nu' \rangle$ for some $\nu' \sim_K \nu$. The run of length $0$ is just the initial state $\langle q_0, \nu_0 \rangle$ (which is a prefix of all runs in $H$). To build the prefix of length $j+1$, we need to extend the prefix of length $j$. We have infinitely many runs whose $j$-th state is $\langle q, \nu' \rangle$ for some $\nu' \sim_K \nu$. We pick an infinite subset of these runs such that the valuations in their $j+1$-st states are equivalent with respect to $\sim_K$. There is always such an infinite subset, because $\sim_K$ has a finite index in $\epsilon$-sampled semantics. We pick a state $\langle q', \nu' \rangle$ such that it can be reached from $\langle q, \nu \rangle$ and it is equivalent with respect to $\sim_K$ to the states in the infinite subset as the $j+1$-st state. Clearly, there is such a state.

For the other direction, let us assume that for each $B$ there is $w_B \in \Sigma^\omega$ such that $w_B\notin L_B^\omega(R)$. We show that $A$ cannot be $\omega$-sampled. For each $\epsilon$ we pick $B = 1/\epsilon$. There is a counter $C_{xy}$ which exceeds $B$ in each accepting run of $R$ over $w_B$. From Lemma~\ref{Lemma:cost-is-real}, each accepting run of $A$ over $w_B$ requires $\numodels{xy} \geq B\cdot \epsilon = 1$ in some state $\langle q, \nu \rangle$ along this run. But from the definition, $\numodels{xy}$ is always strictly smaller than $1$.\qed

}

Note that if a timed automaton can be sampled then one can also compute a sampling rate $\epsilon$. First, it is possible to determine a limit $B$ for the extended R-automaton $R'$ constructed according to Section~\ref{Sec:Encoding-TA-RC} such that $L_B(R') = L(R')$ or $L_B^\omega(R') = L^\omega(R')$. If we know that the language of $R'$ is limited then this can be done by checking the language equality systematically for all values of $B$. Having a value for $B$, we set $\epsilon$ to be equal to $1/(4\cdot B)$. One can also compute a value for $B$ directly from the parameters of $R'$, which is shown in~\cite{aky08r-automata}.

\section{Conclusions}
\label{Sec:Conclusions}


Timed automata with dense time semantics can enforce behaviors, where time distances between events monotonically grow while being bounded by some integer. We have formulated a property distinguishing timed automata which do not use this ability: the untimed language of an automaton in question can be accepted in a semantics where all time delays are multiples of a fixed rational number. These automata preserve all qualitative behaviors (untimed words) when implemented on a platform with a fixed sampling rate. We have also shown that it is decidable whether a timed automaton enjoys this property. The proof characterizes the time differences enforced along runs by a new type of counter automata -- Extended R-automata. As a technical contribution of its own interest, we have shown that limitedness is decidable for these automata.

In spite of this positive outcome, our results show a high degree of complexity present in dense time behaviors enforced by strict inequalities. Therefore, when we require from our model that it can be turned into a sampled implementation, we have to consider usage of strict inequalities with a great care. It is questionable whether the modeling advantages of strict inequalities outweigh the costs of sampling analysis.

\forget{
It is possible to avoid this type of complex behaviors by disallowing strict inequalities. Closed timed automata preserve all qualitative behaviors with sampling rate equal to $1$. These two facts together put in doubt the argument for dense time semantics of timed automata, saying that one does not have to consider a concrete sampling rate of the implementation in the modeling and verification phase.
}

\forget{

\begin{figure}[tbhp]
  \begin{center}

    \[ \xymatrix{
      \ar[r] & *++[o][F=]{l_0} \ar@/^.5cm/[rrrr]^{a,y=1, y:=0} & & & & *++[o][F-]{l_1}
      \ar@/^.5cm/[llll]^{b,1>y \wedge x>1, x:=0} \\
    }
    \]

    \caption{A timed automaton which does not preserve qualitative behaviors in sampled semantics. It enforces the difference between clock values to grow. If the values of $x,y$ are $0.1, 0.6$, respectively, in the location $l_0$ then the difference between the clock values in the location $l_0$ after reading $ab$ will be strictly greater than $0.5$. This example is adapted from~\cite{AlurM04}.}
    \label{Fig:no-sampling}
  \end{center}
\end{figure}

}

\section*{Acknowledgements} We would like to thank Radek Pel\'{a}nek for fruitful discussions and anonymous reviewers for their constructive comments.

\bibliographystyle{alpha}
\bibliography{sampled-semantics}


\end{document}